\documentclass[journal]{IEEEtran}

\usepackage{graphicx,color}   
\usepackage[cmex10]{amsmath}  
\usepackage{amssymb}
\usepackage{cite}  
\usepackage[utf8]{inputenc}
\usepackage[T1]{fontenc}

\newtheorem{theorem}{Theorem}
\newtheorem{remark}{Remark}
\newtheorem{lemma}{Lemma}
\newtheorem{definition}{Definition}

\newtheorem{corollary}{Corollary}

\begin{document}

\title{The Capacity Region of Multiway Relay Channels Over Finite Fields with Full Data Exchange}
\author{Lawrence Ong, Sarah J. Johnson, and Christopher M. Kellett
}

\maketitle

\begin{abstract}
The multi-way relay channel is a multicast network where $L$ users exchange data through a relay. In this paper, the capacity region of a class of multi-way relay channels is derived, where the channel inputs and outputs take values over finite fields. The cut-set upper bound to the capacity region is derived and is shown to be achievable by our proposed functional-decode-forward coding strategy. More specifically, for the general case where the users can transmit at possibly different rates, functional-decode-forward, combined with rate splitting and joint source-channel decoding, is proved to achieve the capacity region; while for the case where all users transmit at a common rate, rate splitting and joint source-channel decoding are not required to achieve the capacity. That the capacity-achieving coding strategies do not utilize the users' received signals in the users' encoding functions implies that feedback does not increase the capacity region of this class of multi-way relay channels.
\end{abstract}


\section{Introduction}
We consider the multi-way relay channel (MWRC), where $L$ users ($L \geq 2$) exchange data via a relay. Each user is to send its data to all other users. We further consider the case where there is no direct link among the users. So, information exchange among the users can only be done through the relay. Common applications of this model include conference calls in the cellular network where mobile users communicate among themselves through a base station, and satellite communications (see Fig.~\ref{fig:mwrc-example}).

The MWRC is an extension of the two-way relay channel (TWRC) where two users exchange data via a relay (e.g., see \cite{knopp06,rankovwittneben06,rankovwittneben07}). As the TWRC embeds a relay channel, coding strategies designed for the relay channel were modified and attempted on the TWRC. These include:
\begin{itemize}
\item \emph{Complete-decode-forward}\footnote{This strategy is commonly referred to as decode-forward or decode-and-forward. We refer to this strategy as complete-decode-forward to differentiate it from our proposed functional-decode-forward} (CDF): The relay completely decodes the users' messages, and broadcasts them back to the users (see \cite{rankovwittneben06,knopp06,rankovwittneben07}).
\item \emph{Compress-forward}: The relay quantizes its received signals, re-encodes and broadcasts them to the users (see \cite{rankovwittneben06,schnurroechtering07}).
\item \emph{Amplify-forward}: The relay simply scales and forwards what it receives (see \cite{rankovwittneben06,knopp06,rankovwittneben07}). When applied to the Gaussian TWRC, this strategy is also known as \emph{analog network coding}~\cite{kattigollakota07}.
\item Combinations of the above strategies (see \cite{schnurrstanczak08,liutao09}).
\item A combination of \emph{partial-decode-forward} and compress-forward (see \cite{gunduztuncel08}).
\end{itemize}

CDF, compress-forward, and amplify-forward coding strategies for the TWRC have been extended to the Gaussian MWRC by  G\"und\"uz \emph{et al.}~\cite{gunduzyener09}. However, none of these strategies achieve the capacity region of the MWRC in general.

\begin{figure}[t]
\centering
\includegraphics[width=5cm]{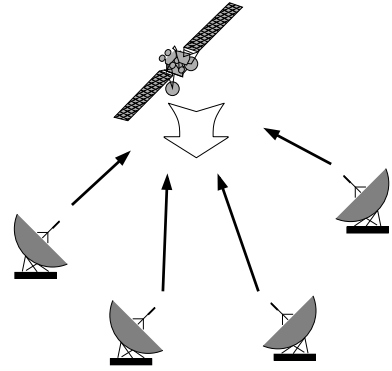}
\caption{An application of the MWRC, where stations exchange information via a satellite}
\label{fig:mwrc-example}
\end{figure}

\subsection{Functional-Decode-Forward}

Recently, \emph{functional-decode-forward} (FDF) has been proposed for the TWRC, where the relay decodes a function of the two users' messages and broadcasts the function back to the users~\cite{knopp07,narayananwilson07,namchung08,namchunglee09,wilsonnarayananpfisersprintson10}. Obviously, the function must be defined such that each user can decode the message of the other user from the function and its own message.
FDF was shown to achieve: (i) the capacity region of the binary TWRC~\cite{knopp07}, where the channels are binary symmetric, and (ii) within $\frac{1}{2}$ bit of the capacity region of the Gaussian TWRC~\cite{namchunglee09}. Linear codes are used in FDF for the binary channel, and lattice codes~\cite{erezzamir04} are used in FDF for the Gaussian channel. FDF for the Gaussian TWRC was extended to the multi-pair Gaussian TWRC (where multiple source-destination pairs exchange data via one relay) by G\"und\"uz \emph{et al.}~\cite{gunduzyener09}.

In the TWRC and the multi-pair TWRC, FDF was designed for pair-wise data exchange. We later proposed FDF for the MWRC (a non-trivial extension of FDF for the TWRC) where multiple users exchange data via a relay at a \emph{common rate}, and showed that FDF achieves the common-rate capacity of the binary MWRC~\cite{ongjohnsonkellett10cl}.
Applying insights from the binary MWRC has allowed us to obtain the common-rate capacity of the Gaussian MWRC with three or more users where all nodes transmit at the same power~\cite{ong10amwrc}.

In this paper, we extend our proposed FDF for the common-rate binary MWRC~\cite{ongjohnsonkellett10cl} to the \emph{general-rate} MWRC over a finite field where the channel inputs and outputs take values over a finite field and where the users can possibly transmit at different rates. Furthermore, unlike \cite{gunduzyener09,ongjohnsonkellett10cl}, we consider the more general \emph{unrestricted} MWRC where each user's encoding function at any time can depend on its own message and its previously received signals.  Note that the binary MWRC is a special case of the MWRC over a finite field.

On the \emph{uplink} (the channel from the users to the relay), we use functional decoding combined with rate splitting. Similar to \cite{ongjohnsonkellett10cl}, linear codes are used here. The main idea behind this generalization (from the binary channel to the finite field channel) relies on the fact that optimal (capacity-achieving) linear codes can be constructed for channels over finite fields. Using linear codes on the uplink, the relay is able to decode a function of the users' codewords, which is also a codeword from the linear code. On the \emph{downlink} (the channel from the relay to the users), the relay needs to send different messages to different users, and so the coding technique for broadcast channels with receiver side information developed by Tuncel \cite{tuncel06} is used, which utilizes joint source-channel decoding. We show that the combination of FDF, rate splitting, and joint source-channel decoding achieves the capacity region of the MWRC over a finite field\footnote{Note that rate splitting and joint source-channel decoding were not required  for the common-rate case in \cite{ongjohnsonkellett10cl}.}.

We shall see later that using the capacity-achieving FDF, the users' transmitted signals only depend on their respective messages and do not depend on their received signals. This means utilizing \emph{feedback} at the users does not increase the capacity region of the MWRC over a finite field.

This, to the best of our knowledge, is the first example of an MWRC where the capacity region is found for all noise distributions/levels. The optimal coding strategy for the MWRC over a finite field proposed in this paper gives insights into optimal processing/coding strategies for other classes of MWRCs.
This work suggests that for the general MWRC, functional decoding should be performed at the relay, and joint source-channel decoding at the users.

On the uplink of MWRCs, the relay receives interfering signals from all the users (see \eqref{eq:def-uplink}). Such networks, where some node(s) receives a function (which can be noisy) of more than one other node's transmission, are usually referred to as \emph{networks with interference}.  Using our proposed FDF, up to two users are allowed to transmit at any time, and the relay attempts to decode a function of the users' messages. Rather than avoiding interference, this coding strategy embraces it and can thus be viewed as a form of \emph{interference alignment}~\cite{viveckjafar08}.

\begin{remark}
Note that linear codes are also used in other types of networks, including the multicast (one source sending data to multiple destinations) network with interference \cite{nazergastpar07,nazergastpar08eu,namchung09arxiv}, the multiple-access channel where the destination is to decode a linear combination of the sources' messages \cite{nazergastpar07,nazergastpar08eu}, and the multi-source multicast network with no interference \cite{danagowaikarpalanki06}. Linear codes have been shown to be optimal (capacity-achieving) in these networks when the channels are themselves linear. Note that the MWRC is not a special case of these networks as it has multiple sources and multiple destinations, and it incorporates interference in its network model. Furthermore, the coding strategy for the uplink developed in this paper is different from existing strategies.
\end{remark}

\subsection{Other Related Work}

A channel model similar to the finite field channel considered in this paper is the deterministic (noiseless) channel. In the deterministic model, the channel output is the arithmetic summation of the \emph{bit-shifted} channel inputs, and there is no noise. The deterministic model has been used to construct coding strategies and to gain insights for more general channels. This approach has been applied to the multiple-access channel~\cite{avestimehrdiggavitse07}, the broadcast channel~\cite{avestimehrdiggavitse07}, the interference channel~\cite{breslertse08,jafarvishwanath10}, the deterministic TWRC~\cite{avestimehrsezgintse08}, and the deterministic multi-pair TWRC~\cite{avestimehrkhajehnejad09itw}. For the deterministic TWRC and the deterministic multi-pair TWRC, it has been shown that linear coding achieves the capacities, an observation similar to that in this paper for the finite field MWRC.

The MWRC we consider herein, where each user is to decode the messages from all other
users, can be seen as a generalization of the TWRC.  Different extensions of the TWRC include:
\begin{itemize}
\item The multi-pair TWRC where multiple source-destination pairs exchange messages via  one relay \cite{avestimehrkhajehnejad09itw,hogowdasun09itw}. Here, each destination only decodes the message from one source. 
\item The multi-pair TWRC where multiple users exchange messages with a base station via a relay \cite{kimsmida10}. Here, each user sends its message to the base station, and the base station sends different messages to each user.
\item The TWRC with additional private messages from the users to the relay~\cite{hogowdasun09icc,hosun10icc}.
\item The MWRC where the users are separated into different groups and all users in each group exchange messages among themselves\cite{gunduzyener09}.
\end{itemize}

The MWRC has also been studied from the point of view of source coding, where multiple users exchange possibly correlated data via a relay. In the source coding setting, the channel from the users to the relay and that from the relay to the users are assumed to be noiseless. The problem formulation is how many bits the users need to encode their respective messages to be sent to the relay; and after the relay receives these encoded messages, how many bits the relay needs to transmit to the users in order for each user to recover the messages of all other users. The three-user lossless case (where each user perfectly reconstructs the other two users' messages) was studied by Wyner \emph{et al.}~\cite{wynerwolf02}, the two-user lossless case and lossy case (where each user reconstructs the other user's message with a prescribed distortion) was studied by Su and El Gamal~\cite{sugamal10}, and the two-user lossy case with common reconstructions  (where each user must also be able to determine the lossy reconstructed message of the other user) was studied by Timo~\emph{et al.}~\cite{roysubmitted}.

\subsection{Organization}

The rest of the paper is organized as follows. In Sec.~\ref{section:channel-model}, we describe the MWRC over a finite field, define the notation used in this paper, and quote a few lemmas that will be used in the later sections. We derive upper bounds to the capacity region and the common-rate capacity of the MWRC over a finite field in Sec.~\ref{section:upper-bound}. We then construct linear codes over finite fields in Sec.~\ref{sec:fields}, which facilitate functional decoding at the relay. We derive the capacity region of the finite field MWRC in Sec.~\ref{section:general-rate}. In Sec.~\ref{section:case-study}, we use the two-user binary MWRC as an example to analyze why neither CDF nor FDF with separate source-channel decoding achieves the capacity region of the MWRC in general. Sec.~\ref{section:conclusion} concludes the paper.

\section{Channel Model}\label{section:channel-model}

\begin{figure}[t]
\centering
\resizebox{8.5cm}{!}{
\begin{picture}(0,0)%
\includegraphics{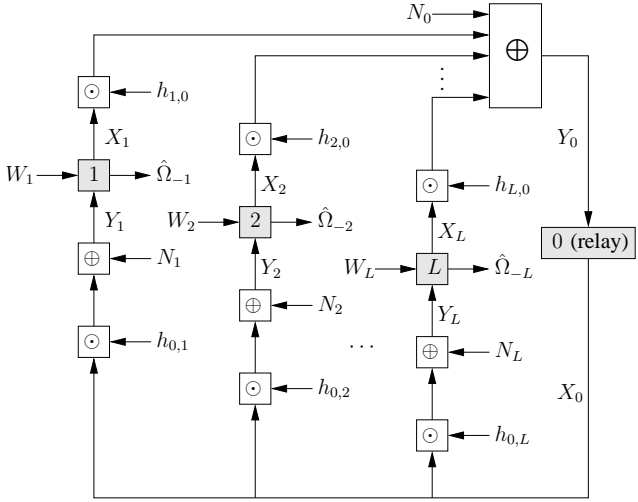}%
\end{picture}%
\setlength{\unitlength}{4144sp}%
\begingroup\makeatletter\ifx\SetFigFont\undefined%
\gdef\SetFigFont#1#2#3#4#5{%
  \fontsize{#1}{#2pt}%
  \fontfamily{#3}\fontseries{#4}\fontshape{#5}%
  \selectfont}%
\fi\endgroup%
\begin{picture}(5472,4311)(301,-3268)
\put(991,-1951){\makebox(0,0)[lb]{\smash{{\SetFigFont{12}{14.4}{\familydefault}{\mddefault}{\updefault}{\color[rgb]{0,0,0}$\odot$}%
}}}}
\put(1621,-1951){\makebox(0,0)[lb]{\smash{{\SetFigFont{12}{14.4}{\familydefault}{\mddefault}{\updefault}{\color[rgb]{0,0,0}$h_{0,1}$}%
}}}}
\put(5041,-1096){\makebox(0,0)[lb]{\smash{{\SetFigFont{12}{14.4}{\familydefault}{\mddefault}{\updefault}{\color[rgb]{0,0,0}$0$  (relay)}%
}}}}
\put(991,-1231){\makebox(0,0)[lb]{\smash{{\SetFigFont{12}{14.4}{\familydefault}{\mddefault}{\updefault}{\color[rgb]{0,0,0}$\oplus$}%
}}}}
\put(1621,-1231){\makebox(0,0)[lb]{\smash{{\SetFigFont{12}{14.4}{\familydefault}{\mddefault}{\updefault}{\color[rgb]{0,0,0}$N_1$}%
}}}}
\put(2386,-1636){\makebox(0,0)[lb]{\smash{{\SetFigFont{12}{14.4}{\familydefault}{\mddefault}{\updefault}{\color[rgb]{0,0,0}$\oplus$}%
}}}}
\put(3016,-1636){\makebox(0,0)[lb]{\smash{{\SetFigFont{12}{14.4}{\familydefault}{\mddefault}{\updefault}{\color[rgb]{0,0,0}$N_2$}%
}}}}
\put(3916,-2761){\makebox(0,0)[lb]{\smash{{\SetFigFont{12}{14.4}{\familydefault}{\mddefault}{\updefault}{\color[rgb]{0,0,0}$\odot$}%
}}}}
\put(4546,-2761){\makebox(0,0)[lb]{\smash{{\SetFigFont{12}{14.4}{\familydefault}{\mddefault}{\updefault}{\color[rgb]{0,0,0}$h_{0,L}$}%
}}}}
\put(316,-511){\makebox(0,0)[lb]{\smash{{\SetFigFont{12}{14.4}{\familydefault}{\mddefault}{\updefault}{\color[rgb]{0,0,0}$W_1$}%
}}}}
\put(1711,-916){\makebox(0,0)[lb]{\smash{{\SetFigFont{12}{14.4}{\familydefault}{\mddefault}{\updefault}{\color[rgb]{0,0,0}$W_2$}%
}}}}
\put(2431,-916){\makebox(0,0)[lb]{\smash{{\SetFigFont{12}{14.4}{\familydefault}{\mddefault}{\updefault}{\color[rgb]{0,0,0}$2$}%
}}}}
\put(1036,-511){\makebox(0,0)[lb]{\smash{{\SetFigFont{12}{14.4}{\familydefault}{\mddefault}{\updefault}{\color[rgb]{0,0,0}$1$}%
}}}}
\put(3241,-1321){\makebox(0,0)[lb]{\smash{{\SetFigFont{12}{14.4}{\familydefault}{\mddefault}{\updefault}{\color[rgb]{0,0,0}$W_L$}%
}}}}
\put(3961,-1321){\makebox(0,0)[lb]{\smash{{\SetFigFont{12}{14.4}{\familydefault}{\mddefault}{\updefault}{\color[rgb]{0,0,0}$L$}%
}}}}
\put(3916,-2041){\makebox(0,0)[lb]{\smash{{\SetFigFont{12}{14.4}{\familydefault}{\mddefault}{\updefault}{\color[rgb]{0,0,0}$\oplus$}%
}}}}
\put(4546,-2041){\makebox(0,0)[lb]{\smash{{\SetFigFont{12}{14.4}{\familydefault}{\mddefault}{\updefault}{\color[rgb]{0,0,0}$N_L$}%
}}}}
\put(2386,-2356){\makebox(0,0)[lb]{\smash{{\SetFigFont{12}{14.4}{\familydefault}{\mddefault}{\updefault}{\color[rgb]{0,0,0}$\odot$}%
}}}}
\put(3016,-2356){\makebox(0,0)[lb]{\smash{{\SetFigFont{12}{14.4}{\familydefault}{\mddefault}{\updefault}{\color[rgb]{0,0,0}$h_{0,2}$}%
}}}}
\put(1171,-916){\makebox(0,0)[lb]{\smash{{\SetFigFont{12}{14.4}{\familydefault}{\mddefault}{\updefault}{\color[rgb]{0,0,0}$Y_1$}%
}}}}
\put(2521,-1321){\makebox(0,0)[lb]{\smash{{\SetFigFont{12}{14.4}{\familydefault}{\mddefault}{\updefault}{\color[rgb]{0,0,0}$Y_2$}%
}}}}
\put(4051,-1726){\makebox(0,0)[lb]{\smash{{\SetFigFont{12}{14.4}{\familydefault}{\mddefault}{\updefault}{\color[rgb]{0,0,0}$Y_L$}%
}}}}
\put(4141,254){\rotatebox{90.0}{\makebox(0,0)[lb]{\smash{{\SetFigFont{12}{14.4}{\familydefault}{\mddefault}{\updefault}{\color[rgb]{0,0,0}$\dotsm$}%
}}}}}
\put(3781,884){\makebox(0,0)[lb]{\smash{{\SetFigFont{12}{14.4}{\familydefault}{\mddefault}{\updefault}{\color[rgb]{0,0,0}$N_0$}%
}}}}
\put(4636,569){\makebox(0,0)[lb]{\smash{{\SetFigFont{12}{14.4}{\familydefault}{\mddefault}{\updefault}{\color[rgb]{0,0,0}$\bigoplus$}%
}}}}
\put(1171,-196){\makebox(0,0)[lb]{\smash{{\SetFigFont{12}{14.4}{\familydefault}{\mddefault}{\updefault}{\color[rgb]{0,0,0}$X_1$}%
}}}}
\put(2521,-601){\makebox(0,0)[lb]{\smash{{\SetFigFont{12}{14.4}{\familydefault}{\mddefault}{\updefault}{\color[rgb]{0,0,0}$X_2$}%
}}}}
\put(4051,-1006){\makebox(0,0)[lb]{\smash{{\SetFigFont{12}{14.4}{\familydefault}{\mddefault}{\updefault}{\color[rgb]{0,0,0}$X_L$}%
}}}}
\put(5086,-196){\makebox(0,0)[lb]{\smash{{\SetFigFont{12}{14.4}{\familydefault}{\mddefault}{\updefault}{\color[rgb]{0,0,0}$Y_0$}%
}}}}
\put(5086,-2401){\makebox(0,0)[lb]{\smash{{\SetFigFont{12}{14.4}{\familydefault}{\mddefault}{\updefault}{\color[rgb]{0,0,0}$X_0$}%
}}}}
\put(3286,-1996){\makebox(0,0)[lb]{\smash{{\SetFigFont{12}{14.4}{\familydefault}{\mddefault}{\updefault}{\color[rgb]{0,0,0}$\dotsm$}%
}}}}
\put(991,209){\makebox(0,0)[lb]{\smash{{\SetFigFont{12}{14.4}{\familydefault}{\mddefault}{\updefault}{\color[rgb]{0,0,0}$\odot$}%
}}}}
\put(2386,-196){\makebox(0,0)[lb]{\smash{{\SetFigFont{12}{14.4}{\familydefault}{\mddefault}{\updefault}{\color[rgb]{0,0,0}$\odot$}%
}}}}
\put(3916,-601){\makebox(0,0)[lb]{\smash{{\SetFigFont{12}{14.4}{\familydefault}{\mddefault}{\updefault}{\color[rgb]{0,0,0}$\odot$}%
}}}}
\put(1621,209){\makebox(0,0)[lb]{\smash{{\SetFigFont{12}{14.4}{\familydefault}{\mddefault}{\updefault}{\color[rgb]{0,0,0}$h_{1,0}$}%
}}}}
\put(3016,-196){\makebox(0,0)[lb]{\smash{{\SetFigFont{12}{14.4}{\familydefault}{\mddefault}{\updefault}{\color[rgb]{0,0,0}$h_{2,0}$}%
}}}}
\put(4546,-601){\makebox(0,0)[lb]{\smash{{\SetFigFont{12}{14.4}{\familydefault}{\mddefault}{\updefault}{\color[rgb]{0,0,0}$h_{L,0}$}%
}}}}
\put(1621,-511){\makebox(0,0)[lb]{\smash{{\SetFigFont{12}{14.4}{\familydefault}{\mddefault}{\updefault}{\color[rgb]{0,0,0}$\hat{\Omega}_{-1}$}%
}}}}
\put(3016,-916){\makebox(0,0)[lb]{\smash{{\SetFigFont{12}{14.4}{\familydefault}{\mddefault}{\updefault}{\color[rgb]{0,0,0}$\hat{\Omega}_{-2}$}%
}}}}
\put(4546,-1321){\makebox(0,0)[lb]{\smash{{\SetFigFont{12}{14.4}{\familydefault}{\mddefault}{\updefault}{\color[rgb]{0,0,0}$\hat{\Omega}_{-L}$}%
}}}}
\end{picture}%
}
\caption{The $L$-user MWRC over a finite field $\mathcal{F}$ with associated addition $\oplus$ and multiplication $\odot$, where $\hat{\Omega}_{-i} \triangleq(\hat{W}_{i,1}, \dotsc, \hat{W}_{i,i-1}, \hat{W}_{i,i+1}, \dotsc, \hat{W}_{i,L})$ is user $i$'s estimate of all other users' messages}
\label{fig:mwrc}
\end{figure}

Fig.~\ref{fig:mwrc} depicts the $L$-user MWRC considered in this paper, where there is no direct user-to-user link. Nodes 1, 2, $\dotsc$, $L$ are the users, and node $0$ the relay. By definition, $L \geq 2$, and each user is to decode the messages from all other users, i.e., the users perform \emph{full data exchange}.  We denote by $X_i$ node $i$'s input to the channel, $Y_i$ the channel output received by node $i$, and $W_i$ node $i$'s message. We assume that the messages are independent. We consider a  full-duplex and causal relay, meaning that the relay can transmit and receive at the same time, and that the transmit signal of the relay at any time can only depend on its past received signals.

\begin{definition}\label{definition:channel}
We define the $L$-user MWRC over a finite field $\mathcal{F}$ (with associated addition $\oplus$, multiplication $\odot$, and the additive identity $\mathfrak{0} \in \mathcal{F}$) as follows:
\begin{itemize}
\item The uplink channel is the \emph{weighted} sum of all users' channel inputs and the relay's receiver noise:
\begin{subequations}
\begin{align}
Y_0 &= \left(\bigoplus_{i=1}^L (h_{i,0} \odot X_i) \right) \oplus N_0\\
& \triangleq (h_{1,0} \odot X_1) \oplus (h_{2,0} \odot X_2) \oplus \dotsm \oplus (h_{L,0} \odot X_L) \nonumber\\
&\quad  \oplus N_0, \label{eq:def-uplink}
\end{align}
\end{subequations}
where $X_i, N_0, Y_0 \in \mathcal{F}$, $h_{i,0} \in \mathcal{F}\setminus \{\mathfrak{0}\}$, $\forall i$, and $N_0$ is the receiver noise and is an independent and identically distributed (i.i.d.) random variable for each channel use. The parameters $h_{i,0}$, $\forall i$, are fixed and are known to all the nodes \emph{a priori}. Recall that $\mathcal{F}$ is a field if and only if $|\mathcal{F}|=\ell^z$ for some prime number $\ell$ and some positive integer $z$. 
\item The downlink consists of independent channels from the relay to the users:
\begin{equation}
Y_i = (h_{0,i} \odot X_0) \oplus N_i, \quad \forall i \in \{1,2,\dotsc,L\},
\end{equation}
where $X_0, N_i, Y_i \in \mathcal{F}$, $h_{0,i} \in \mathcal{F} \setminus \{\mathfrak{0}\}$, $\forall i$, and $N_i$ is the receiver noise at node $i$ and is an i.i.d. random variable for each channel use and for each user $i$. Each $h_{0,i}$ is fixed for all channel uses and is known to node $i$ \emph{a priori}. 
\end{itemize}
\end{definition}

\begin{remark}
The MWRC over a finite field is defined to resemble the wireless additive white Gaussian noise channel where the channel output is the sum of attenuated (usually as a result of path loss, which is inversely proportional to the node distances) channel inputs and noise. However, addition and multiplication over a field do not bear the same practical implication as those over real numbers.
\end{remark}

Let $X_i[t]$ and $Y_i[t]$ denote the transmitted signal and the received signal of user $i$ respectively on the $t$-th channel use.
We consider the following block code of $n$ simultaneous uplink and downlink channel uses, meaning that the relay and all users transmit $X_i[t]$ respectively and simultaneously, for $t \in \{1,2,\dotsc,n\}$.
\begin{definition}
A $(2^{nR_1},2^{nR_2}, \dotsc, 2^{nR_L}, n)$ code for the MWRC consists of
\begin{enumerate}
\item $L$ messages, one for each user: $W_i \in \mathcal{W}_i = \{1,\dotsc,2^{nR_i}\}$, for $i \in \{1,2,\dotsc, L\}$. We denote by $\Omega \triangleq (W_1,W_2,\dotsc,W_L)$ the \emph{message tuple}. 
\item $L$ sets of user encoding functions, one set for each user: $f_{i,t}: \mathcal{W}_i \times \mathcal{F}^{t-1} \rightarrow \mathcal{F}$, such that $X_i[t] = f_{i,t}(W_i,Y_i[1],Y_i[2],\dotsc,Y_i[t-1])$, for  $i \in \{1,2,\dotsc,L\}$, $t \in \{1,2,\dotsc,n\}$. This means that the transmit signal of a user at any time can depend on its message and its previously received signals.
\item A set of relay encoding functions: $f_{0,t}: \mathcal{F}^{t-1} \rightarrow \mathcal{F}$, for $t \in \{1,2,\dotsc,n\}$, such that  $X_0[t] = f_{0,t}(Y_0[1],Y_0[2], \dots, Y_0[t-1])$. This means the transmit signal of the relay at any time can only depend on its previously received signals.
\item $L$ user decoding functions, one for each user: $g_i: \mathcal{F}^n \times \mathcal{W}_i \rightarrow \mathcal{W}_1 \times \dotsm \times \mathcal{W}_{i-1} \times \mathcal{W}_{i+1} \times \dotsm \times \mathcal{W}_L$, such that $\hat{\Omega}_{-i} \triangleq(\hat{W}_{i,1}, \dotsc, \hat{W}_{i,i-1}, \hat{W}_{i,i+1}, \dotsc, \hat{W}_{i,L}) = g_i(\boldsymbol{Y}_i,W_i)$, for $i \in \{1,2,\dotsc,L\}$, where $\hat{W}_{i,j}$ is node $i$'s estimate of $W_j$, and $\boldsymbol{Y}_i = (Y_i[1], Y_i[2],\dotsc,Y_i[n])$. This means each user decodes the messages sent by all other users based on its $n$ received signals and the knowledge of its own message.
\end{enumerate}
\end{definition}

Note that the source message $W_i$, which is an $nR_i$-bit message, is sent from user $i$ to all other nodes (through the relay) in $n$ channel uses, giving a rate of $\frac{nR_i}{n}=R_i$ bits/channel use. We say that user $i$ transmits at the rate $R_i$ bits/channel use.

In this paper, bold letters are used to denote collections of variables across time, e.g., $\boldsymbol{X} = (X[1], X[2],$ $\dotsc, X[k])$, for some integer $k > 1$. The length of the vector will be explicitly mentioned when it is not clear from the context. For a random variable $X$, we use the corresponding lower case $x$ to denote its realization. 

\begin{definition}
Assuming that the message tuple $\Omega \triangleq (W_1,W_2,\dotsc,W_L)$ is uniformly distributed over the product set $\mathfrak{W} \triangleq \mathcal{W}_1 \times \mathcal{W}_2 \times \dotsm \times \mathcal{W}_L$, the \emph{average error probability} for the $(2^{nR_1},2^{nR_2},\dotsc,2^{nR_L},n)$ code is defined as
\begin{subequations}
\begin{align}
P_\text{e} &= \Pr \Big\{\hat{W}_{i,j} \neq W_j, \text{ for some } j \in \{1,2,\dotsc,L\} \nonumber\\
&\quad\quad\quad \text{ and some } i \in\{1,2,\dotsc,L\} \setminus j\Big\} \\
&= \frac{1}{2^{n\sum_{j=1}^L R_j}} \sum_{\omega \in \mathfrak{W}} \Pr \left\{ \bigcup_{i=1}^L \Big\{ \hat{\Omega}_{-i} \neq \omega_{-i}\Big\} \Bigg\vert \Omega = \omega \right\},
\end{align}
\end{subequations}
where $\omega_{-i} = (w_1,\dotsc,w_{i-1},w_{i+1},\dotsc,w_L)$ is defined as $\omega$ without the $i$-th entry.
\end{definition}

\begin{definition}\label{definition:rate}
A \emph{rate tuple} $(R_1,R_2, \dotsc, R_L)$ is said to be \emph{achievable} if, for any $\epsilon > 0$, there is at least one $(2^{nR_1},2^{nR_2},\dotsc,2^{nR_L},n)$ code such that $P_\text{e} < \epsilon$.
\end{definition}

We say that a node can \emph{reliably} decode a message if and only if the average probability that the node wrongly decodes the message can be made arbitrarily small. Hence, the rate tuple $(R_1,R_2, \dotsc, R_L)$ is achievable if each user can reliably decode the messages from all other users.

\begin{definition}\label{definition:capacity}
The \emph{capacity} region $\mathcal{C}$ is defined as the closure of all achievable rate tuples.
\end{definition}

In this paper, we also consider the common-rate case (a special case) where all users transmit at $R = R_i$, $\forall i \in \{1,2,\dotsc,L\}$. We say that the common rate $R$ is achievable if the rate tuple $(R,R,\dotsc,R)$ is achievable. The \emph{common-rate} capacity can be similarly defined:
\begin{definition}
We define the \emph{common-rate capacity} (also known as the symmetrical capacity~\cite{gunduzyener09}) as
\begin{equation}
C \triangleq \sup \{R: (R,R,\dotsc,R) \text{ is achievable}\}.
\end{equation}
\end{definition}

The common rate is useful in systems where all users have the same amount of information to send, or in \emph{fair} systems where every user is to be given the same guaranteed uplink \emph{bandwidth}, i.e., each user can send data up to a certain rate, at which all other users are able to decode.

To simplify equations in this paper, we define
\begin{align}
R_\text{min} &= \min\limits_{j\in \{1,2,\dotsc,L\}} R_j\\
R_i^{\text{c}} &= \left(\sum\limits_{j=1}^LR_j \right) - R_i\\
R_\text{min}^{\text{c}} &= \left( \sum\limits_{j=1}^L R_j \right) - R_\text{min}.
\end{align}

For a random variable $X \in \mathcal{X}$, $H(X) = - \sum_{x \in \mathcal{X}} p(x)\log_2 p(x)$ is the entropy of $X$. We denote the uniform distribution of $X$ by $p^{\text{u}}(x)$.

\subsection{Existing Results}

In this section, we quote existing results that will be used in the later sections in this paper.

First, for a finite field $\mathcal{F}$ with associated operations of addition $\oplus$, multiplication $\odot$, and the additive identity $\mathfrak{0} \in \mathcal{F}$, we have the following lemma due to Jelinek~\cite[Lemma 9.3]{jelinek68}:
\begin{lemma}\label{lemma:jelinek}
Consider a finite field $\mathcal{F}$. We have the following
\begin{enumerate}
\item the equation $a \oplus x = b$ (where $x$ is the unknown) has a unique solution in
$\mathcal{F}$,
\item for each $a \in \mathcal{F}$, the set $\{a \oplus x:x \in \mathcal{F}\}$ is equal to $\mathcal{F}$.
\item the equation $c \odot y = d$ (where $y$ is the unknown) has a unique solution in $\mathcal{F}$ provided $c \neq \mathfrak{0}$.
\item for each $c \in \mathcal{F} \setminus \{\mathfrak{0}\}$, the set $\{ c \odot y: y \in \mathcal{F}\}$ is equal to $\mathcal{F}$.
\end{enumerate}
\end{lemma}

In this paper, we prove achievability and capacity results based on the properties of the set of jointly $\delta$-typical sequences, which is defined as follows:
\begin{definition}\label{def:typical}
The jointly $\delta$-typical set $\mathcal{A}^n_{[XY]\delta}$ with respect to a distribution $p(x,y)$ on $\mathcal{X}\times\mathcal{Y}$ is the set of sequences $(\boldsymbol{x},\boldsymbol{y})= ( (x_1,y_1),(x_2,y_2),\dotsc,(x_n,y_n) ) \in \mathcal{X}^n \times \mathcal{Y}^n$ such that
\begin{align}
\left\vert -\frac{1}{n} \log_2 p(\boldsymbol{x}) - H(X) \right\vert &< \delta\\
\left\vert -\frac{1}{n} \log_2 p(\boldsymbol{y}) - H(Y) \right\vert &< \delta\\
\left\vert -\frac{1}{n} \log_2 p(\boldsymbol{x},\boldsymbol{y}) - H(X,Y) \right\vert &< \delta,
\end{align}
where $p(\boldsymbol{x},\boldsymbol{y}) = \prod_{i=1}^np(x_i,y_i)$.
The sequences in $\mathcal{A}^n_{[XY]\delta}$ are called jointly $\delta$-typical sequences.
\end{definition}

The jointly $\delta$-typical set has the following properties (taken from \cite[pages 196--197]{coverthomas06}):
\begin{lemma}\label{lemma:jaep}
Let 
\begin{equation}
(\boldsymbol{X},\boldsymbol{Y})=((X_1,Y_1),(X_2,Y_2),\dotsc,(X_n,Y_n)),
\end{equation}
where $(X_i,Y_i)$ are i.i.d. drawn according to $p(x,y)$.
The following holds for sufficiently large $n$:
\begin{equation}
\Pr\left\{ (\boldsymbol{X},\boldsymbol{Y}) \in \mathcal{A}^n_{[XY]\delta} \right\} > 1 - \delta.
\end{equation}
\end{lemma}

\begin{lemma}\label{lemma:jaep-2}
Let $(\tilde{\boldsymbol{X}},\tilde{\boldsymbol{Y}})=((\tilde{X}_1,\tilde{Y}_1),\dotsc,(\tilde{X}_n,\tilde{Y}_n))$ where $(\tilde{X}_i,\tilde{Y}_i)$ are i.i.d. drawn according to $p(x)p(y)$ (where $p(x)$ and $p(y)$ are the marginal probability distribution functions of $p(x,y)$). Then,
\begin{equation}
\Pr \left\{ (\tilde{\boldsymbol{X}},\tilde{\boldsymbol{Y}}) \in \mathcal{A}^n_{[XY]\delta} \right\} \leq 2^{-n(I(X;Y)-3\delta)}.
\end{equation}
\end{lemma}

Next, we have the following theorem due to Tuncel~\cite{tuncel06} for the broadcast channel with receiver side information.
\begin{theorem}\label{theorem:joint-source-channel-coding}
Consider a broadcast channel $p(y_1,y_2,\dotsc,y_L|x_0)$ where node 0 is the source and nodes 1, 2, $\dotsc$, $L$ are receivers. Node 0 is to send a message $\boldsymbol{U}= (U^{(1)},U^{(2)}\dotsc,U^{(n_s)})$ to all the receivers, and each receiver $i$ has side information $\boldsymbol{S}_i= (S_i^{(1)},S_i^{(2)},\dotsc,S_i^{(n_s)})$ \emph{a priori}. Each $(U^{(v)},S_1^{(v)},S_2^{(v)},\dotsc,S_L^{(v)})$ is i.i.d. according to $p(u,s_1,s_2,\dotsc,s_L)$, for all $v \in \{1,2,\dotsc,n_s\}$. The source transmits $\boldsymbol{X}_0(\boldsymbol{U})$ as a function of $\boldsymbol{U}$ in $n$ channel uses. Each receiver $i$ can reliably decode $\boldsymbol{U}$, from its $n$ received channel outputs $\boldsymbol{Y}_i$ and its side information $\boldsymbol{S}_i$, if $n_s$ and $n$ are sufficiently large and if
\begin{equation}
H(U|S_i) <  \frac{n}{n_s} I(X_0;Y_i),\quad \forall i \in \{1,2,\dotsc,L\}, \label{eq:tuncel}
\end{equation}
for some $p(x_0)$.
\end{theorem}

To show achievability in the Theorem~\ref{theorem:joint-source-channel-coding}, joint source-channel decoding is utilized in the sense that each receiver uses its side information in the channel decoding.

We will use the above result for the downlink of the MWRC in Sec.~\ref{section:general-rate}.  On the downlink, the relay transmits a function of the users' messages that it has decoded on the uplink. Each user $i$ decodes the function sent by the relay from its received symbols and its own message $W_ i$ as side information.

\section{Upper Bounds to The Capacity Region and The Common-Rate Capacity}\label{section:upper-bound}

In this section, we derive cut-set upper bounds to the capacity region and the common-rate capacity of the MWRC over a finite field. A cut-set upper bound to the capacity region of a network is the maximum rate that information can be transferred across a \emph{cut} separating two disjoint sets of nodes, assuming that all nodes on each side of the cut can fully cooperate~\cite[page 591]{coverthomas06}.

\begin{theorem}\label{theorem:general-upper-bound}
Consider the $L$-user MWRC over a finite field $\mathcal{F}$. If the rate tuple $(R_1,R_2,\dotsc,R_L)$ is achievable, then
\begin{align}
R_\text{min}^{\text{c}} &\leq \log_2|\mathcal{F}| - H(N_0) \label{eq:mwrc-ub-1}\\
R_i^{\text{c}} &\leq \log_2|\mathcal{F}| - H(N_i), \quad \forall i \in \{1,2,\dotsc,L\}. \label{eq:mwrc-ub-2}
\end{align}
\end{theorem}

\begin{proof}[Proof of Theorem~\ref{theorem:general-upper-bound}]
Consider a network of $m$ nodes, in which node $i$ sends information at the rate $R_{i,j}$ to node $j$. If the set of rates $\{R_{i,j}\}$  are achievable, there exists some joint probability distribution $p(x_1,x_2,\dotsc,x_m)$ such that the sum rate across a cut is constrained by~\cite[Theorem 15.10.1]{coverthomas06}
\begin{equation}
\sum_{i \in \mathcal{S}, j \in \mathcal{S}^{\text{c}}} R_{i,j} \leq I(X_\mathcal{S};Y_{\mathcal{S}^{\text{c}}} | X_{\mathcal{S}^{\text{c}}} ),
\end{equation}
for all $\mathcal{S} \subset \{1,2,\dotsc,m\}$. Here $X_{\mathcal{S}} = \{X_i: i \in \mathcal{S}\}$, and $\mathcal{S}^{\text{c}} = \{1,2,\dotsc,m\} \setminus \mathcal{S}$.

First, we consider the cut separating $\mathcal{S} = \{1,2,\dotsc,L\} \setminus \{i\}$ for some $i \in \{1,2,\dotsc,L\}$, and $\mathcal{S}^{\text{c}}= \{0,i\}$. The total information flow from $\mathcal{S}$ to $\mathcal{S}^{\text{c}}$ is $(W_1,W_2,\dotsc,W_{i-1},W_{i+1},\dotsc,W_L)$ with the sum rate of $\sum_{j=1, j \neq i}^L R_j = R_i^{\text{c}}$. We have the following rate constraint on $R_i^{\text{c}}$, for each $i \in \{1,2,\dotsc,L\}$:
\begin{subequations}
\begin{align}
R_i^{\text{c}} &\leq I(X_\mathcal{S}; Y_{\mathcal{S}^{\text{c}}} | X_{\mathcal{S}^{\text{c}}} )\\
&=  H(Y_{\mathcal{S}^{\text{c}}}|X_{\mathcal{S}^{\text{c}}}) - H(Y_{\mathcal{S}^{\text{c}}}|X_\mathcal{S},X_{\mathcal{S}^{\text{c}}} )\\
&=  H(Y_0,Y_i|X_0,X_i ) - H(Y_0,Y_i|X_{\{0,1\dotsc,L\}})  \label{eq:ub-1}\\
&= H\left( \left[ \bigoplus_{j\in\mathcal{S}} (h_{j,0} \odot X_j) \right] \oplus N_0,N_i \right) - H(N_0,N_i)\\
&= H\left( \left[ \bigoplus_{j\in\mathcal{S}} (h_{j,0} \odot X_j) \right] \oplus N_0\right) + H(N_i) - H(N_0) \nonumber\\
&\quad- H(N_i) \label{eq:indep-noise}\\
&= H\left( \left[ \bigoplus_{j\in\mathcal{S}} (h_{j,0} \odot X_j) \right] \oplus N_0\right) - H(N_0),\label{eq:binary-ub-1}
\end{align}
\end{subequations}
where \eqref{eq:indep-noise} is because $\left( [\bigoplus_{i\in\mathcal{S}} X_i ]\oplus N_0\right)$ and $N_i$ are statistically independent, so are $N_0$ and $N_i$.

Now, we consider the  cut separating $\mathcal{S} = \{0,1,2,\dotsc,L\} \setminus \{i\}$ for some $i \in \{1,2,\dotsc,L\}$, and $\mathcal{S}^{\text{c}}= \{i\}$. The total information flow from $\mathcal{S}$ to $\mathcal{S}^{\text{c}}$ is again $(W_1,W_2,\dotsc,W_{i-1},W_{i+1},\dotsc,W_L)$ with the sum rate of $R_i^{\text{c}}$. We have the following rate constraint on $R_i^{\text{c}}$, for each $i \in \{1,2,\dotsc,L\}$.
\begin{subequations}
\begin{align}
R_i^{\text{c}}
&\leq I(X_\mathcal{S}; Y_{\mathcal{S}^{\text{c}}} | X_{\mathcal{S}^{\text{c}}} )\\
&=  H(Y_{\mathcal{S}^{\text{c}}}|X_{\mathcal{S}^{\text{c}}}) - H(Y_{\mathcal{S}^{\text{c}}}|X_\mathcal{S},X_{\mathcal{S}^{\text{c}}} ) \\
&= H(Y_i|X_i ) - H(Y_i|X_{\{0,1\dotsc,L\}}) \label{eq:ub-2}\\
&= H((h_{0,1} \odot X_0) \oplus N_i) - H(X_0 \oplus N_i | X_{\{0,1\dotsc,L\}})\\
&= H((h_{0,1} \odot X_0) \oplus N_i) - H(N_i). \label{eq:binary-ub-2}
\end{align}
\end{subequations}

All achievable rate tuples must be bounded by the two constraints \eqref{eq:binary-ub-1} and \eqref{eq:binary-ub-2} for all $i$ and for some $p(x_0,x_1,\dotsc,x_L)$. Note that $H(N_i)$, $\forall i$, only depends on the respective noise distributions and does not depend on the choice of input distribution $p(x_0,x_1,\dotsc,x_L)$.

For any discrete random variable $X \in \mathcal{F}$, the maximum of $H(X)$ is $\log_2|\mathcal{F}|$ and is attained by the uniform distribution $p^{\text{u}}(x)$ \cite[Theorem 2.6.4]{coverthomas06}.
For a random variable $N \in \mathcal{F}$ and a constant $h \in \mathcal{F} \setminus \{\mathfrak{0}\}$, from Lemma~\ref{lemma:jelinek}, there is a bijective (one-to-one and onto) mapping from $X$ to $Y=[(h \odot X) \oplus N]$. So, if $p(x)$ is a uniform distribution, then for any $N=n$, $p(y|n)$ is a uniform distribution. Averaged over all $n$, $p(y) = \sum_{n \in \mathcal{F}} p(y|n)p(n)$ is also a uniform distribution. So, choosing the independent and uniform distribution $p(x_0,x_1,\dotsc,x_L)=p^{\text{u}}(x_0)p^{\text{u}}(x_1)\dotsm p^{\text{u}}(x_L)$ simultaneously maximizes \eqref{eq:binary-ub-1} and \eqref{eq:binary-ub-2} for all $i \in \{0,1,\dotsc,L\}$, giving
\begin{align}
R_i^{\text{c}} &\leq \log_2|\mathcal{F}| - H(N_0) \label{eq:mwrc-ub-4}\\
R_i^{\text{c}} &\leq \log_2|\mathcal{F}| - H(N_i), \label{eq:mwrc-ub-3}
\end{align}
for all $i \in \{1,2,\dotsc,L\}$. Eqn.~\eqref{eq:mwrc-ub-4} can be further simplified to $R_\text{min}^{\text{c}} \triangleq \max_{i \in \{1,2,\dotsc,L\}} R_i^{\text{c}} \leq \log_2|\mathcal{F}| - H(N_0)$. This gives Theorem~\ref{theorem:general-upper-bound}.
\end{proof}

For the common rate case, we have the following upper bound on the common-rate capacity:
\begin{corollary}\label{theorem:upper-bound}
Consider the $L$-user MWRC over a finite field $\mathcal{F}$. The common-rate capacity is upper-bounded by
\begin{equation}
C \leq \frac{1}{L-1}\left(\log_2 |\mathcal{F}| - \max\limits_{i \in \{0,1,\dotsc, L\}} H(N_i)\right).
\end{equation}
\end{corollary}

\begin{proof}[Proof of Corollary~\ref{theorem:upper-bound}]
Under the constraint $R= R_i$, $\forall i\in \{1,2,\dotsc,L\}$, we have $R_\text{min}^{\text{c}} = R_i^{\text{c}} = (L-1)R$, $\forall i$. So, \eqref{eq:mwrc-ub-1} and \eqref{eq:mwrc-ub-2} in Theorem~\ref{theorem:general-upper-bound} simplify to $(L-1)R \leq \log_2 |\mathcal{F}| - H(N_i)$, for $i \in \{0,1,\dotsc,L\}$.
\end{proof}

\section{Fields and Linear Codes} \label{sec:fields}

Random linear codes will be employed by the users to transmit their respective source messages to the relay in the FDF coding strategy. Using random linear codes, for any two  messages the corresponding codewords are statistically independent, and the summation of these two codewords is also a codeword with the same structure and properties as the original codewords. With this, the relay will be able to decode the summation of two codewords to obtain the desired function of the source messages without needing to decode the individual messages.
In this section, we present a construction of \emph{random}
linear codes with elements from finite fields, and prove in Theorem~\ref{theorem:ff-linear-code} that these codes achieve the capacity region of the finite field adder channel.

Consider a message of the form $\boldsymbol{s} \in \mathcal{F}^k$,
and a linear code that maps $\boldsymbol{s}$ to a length-$n$ codeword $\boldsymbol{x} \in \mathcal{F}^n$:
\begin{subequations}
\begin{align}
\boldsymbol{x} &= ( \boldsymbol{s} \odot \mathbb{G} ) \oplus \boldsymbol{q}\label{eq:linear-codes-def-1}\\
&= \left( \boldsymbol{s} \odot \begin{bmatrix} \boldsymbol{g}_1 \\ \boldsymbol{g}_2 \\ \vdots \\ \boldsymbol{g}_k \end{bmatrix} \right) \oplus \boldsymbol{q},\label{eq:linear-codes-def-2}
\end{align}
\end{subequations}
where $\boldsymbol{x}$ is a row vector of length $n$, $\boldsymbol{s}$ is a row vector of length $k$, $\mathbb{G}$ is a fixed $k$-by-$n$ matrix, with each element independently and uniformly chosen over $\mathcal{F}$, $\boldsymbol{g}_i$, the $i$-th row in $\mathbb{G}$, is a row vector of length $n$, and $\boldsymbol{q}$ is a fixed row vector of length $n$, with each element independently and uniformly chosen over $\mathcal{F}$.

We will show that the codeletter of the above code is uniform i.i.d., and any two codewords are independent. We extend Gallager's results for
binary linear codes \cite[pages 206--207]{gallager68} to
finite field linear codes in the following two lemmas.

\begin{lemma}\label{lemma:linear-codes-1}
Consider the linear codes defined in \eqref{eq:linear-codes-def-1}. Over the ensemble of codes, the probability that a message $\boldsymbol{s}_1$ is mapped to a given codeword $\boldsymbol{x}_1$ is $p(\boldsymbol{x}_1)=|\mathcal{F}|^{-n}$.
\end{lemma}

\begin{proof}[Proof of Lemma~\ref{lemma:linear-codes-1}]
There are $|\mathcal{F}|^{n(k+1)}$ ways of selecting $\mathbb{G}$ and $\boldsymbol{q}$. As the elements are arbitrarily chosen, each $(\mathbb{G},\boldsymbol{q})$ has a probability of $|\mathcal{F}|^{-n(k+1)}$ of being selected. Following from Lemma~\ref{lemma:jelinek}, for any $\mathbb{G}$, there is only one $\boldsymbol{q}$ that results in the given $\boldsymbol{x}_1$. So, there are only $|\mathcal{F}|^{nk}$ different $(\mathbb{G},\boldsymbol{q})$'s that map $\boldsymbol{s}_1$ to $\boldsymbol{x}_1$. Hence, $p(\boldsymbol{x}_1) = |\mathcal{F}|^{nk} |\mathcal{F}|^{-n(k+1)} = |\mathcal{F}|^{-n}$.
\end{proof}

\begin{lemma}\label{lemma:linear-codes-2}
Consider the linear codes defined in \eqref{eq:linear-codes-def-1}. Let $\boldsymbol{s}_1$ and $\boldsymbol{s}_2$ be two different messages. The corresponding codewords, i.e.,
\begin{align}
\boldsymbol{x}_1 &= (\boldsymbol{s}_1 \odot \mathbb{G}) \oplus \boldsymbol{q}\\
\boldsymbol{x}_2 &= (\boldsymbol{s}_2 \odot \mathbb{G}) \oplus \boldsymbol{q},
\end{align}
are statistically independent.
\end{lemma}

\begin{proof}[Proof of Lemma~\ref{lemma:linear-codes-2}]
To show independence, we need to find the probabilities $p(\boldsymbol{x}_1)$ and $p(\boldsymbol{x}_2|\boldsymbol{x}_1)$, and show that $p(\boldsymbol{x}_1,\boldsymbol{x}_2) =p(\boldsymbol{x}_1)p(\boldsymbol{x}_2)$. Equivalently, we find the probabilities $p(\boldsymbol{x}_1 \oplus -\boldsymbol{x}_2)$ and $p(\boldsymbol{x}_1|\boldsymbol{x}_1 \oplus -\boldsymbol{x}_2)$, where $-\boldsymbol{x}_2$ is the \emph{additive inverse} of $\boldsymbol{x}_2$ in $\mathcal{F}$.
Let $\boldsymbol{s}_1$ and $\boldsymbol{s}_2$ differ in the $j$-th position (they may differ, additionally, in other positions). So, $\boldsymbol{x}_1 \oplus -\boldsymbol{x}_2 = (\boldsymbol{s}_1 \oplus -\boldsymbol{s}_2) \odot \mathbb{G}$. For any $(\boldsymbol{g}_1,\dotsc,\boldsymbol{g}_{j-1},\boldsymbol{g}_{j+1},\dotsc,\boldsymbol{g}_k)$, there is only one $\boldsymbol{g}_j$ that results in the given $(\boldsymbol{x}_1 \oplus -\boldsymbol{x}_2)$. Hence, there are only $|\mathcal{F}|^{n(k-1)}$ different $\mathbb{G}$'s that give $(\boldsymbol{x}_1 \oplus  - \boldsymbol{x}_2)$. In addition, for any chosen $\mathbb{G}$ that gives the required $(\boldsymbol{x}_1 \oplus -\boldsymbol{x}_2)$, there is only one $\boldsymbol{q}$ that results in the given $\boldsymbol{x}_1$. So, there are only $|\mathcal{F}|^{n(k-1)}$ unique $(\mathbb{G},\boldsymbol{q})$'s that give the desired $(\boldsymbol{x}_1 \oplus -\boldsymbol{x}_2, \boldsymbol{x}_1)$ or equivalently the desired $(\boldsymbol{x}_1,\boldsymbol{x}_2)$. Again each $(\mathbb{G},\boldsymbol{q})$ has a probability of $|\mathcal{F}|^{-n(k+1)}$ of being selected. So, the probability $p(\boldsymbol{x}_1,\boldsymbol{x}_2) = |\mathcal{F}|^{n(k-1)} |\mathcal{F}|^{-n(k+1)} = |\mathcal{F}|^{-2n} = p(\boldsymbol{x}_1)p(\boldsymbol{x}_2)$.
\end{proof}

\begin{remark}
The key in proving Lemma~\ref{lemma:linear-codes-2} is to find the probability of the summation of the first codeword and the additive inverse of the second codeword, rather than the summation of the two codewords (as in the binary case \cite[page 207]{gallager68}). Note that for the binary case, the additive inverse of a codeword is the codeword itself.
\end{remark}

\begin{remark}
Note that although the \emph{dither} vector $\boldsymbol{q}$ is not required for proving that two codewords are independent (Lemma~\ref{lemma:linear-codes-2}), it is required for proving that all codeletters for any codeword are independent and uniformly distributed (Lemma~\ref{lemma:linear-codes-1}). 
\end{remark}

\begin{theorem} \label{theorem:ff-linear-code}
Consider a point-to-point finite field adder channel
\begin{equation}\label{eq:ff-channel}
Y = X \oplus N,
\end{equation}
where $X \in \mathcal{F}$ is the channel input from the transmitter, $Y \in \mathcal{F}$ is the channel output received by the receiver, and $N \in \mathcal{F}$ is the channel noise and is an i.i.d. random variable for each channel use.
Using the linear code in \eqref{eq:linear-codes-def-1}, the source sends a message $\boldsymbol{S}$, which is uniformly distributed in $\mathcal{F}^k$, over $n$ uses of the channel, $\boldsymbol{X}(\boldsymbol{S})$. The receiver can decode the message $\boldsymbol{S}$ from the $n$ received signals $\boldsymbol{Y}$ with arbitrarily small error probability if $n$ is sufficiently large and if
\begin{equation}\label{eq:linear-capacity}
\frac{k\log_2|\mathcal{F}|}{n} < \log_2 |\mathcal{F}| - H(N).
\end{equation}
\end{theorem}

\begin{proof}[Proof of Theorem~\ref{theorem:ff-linear-code}]
The source transmits $\boldsymbol{X}(\boldsymbol{S}) = \left(\boldsymbol{S} \odot \mathbb{G} \right) \oplus \boldsymbol{q}$, according to \eqref{eq:linear-codes-def-1}, over $n$ channel uses.
The receiver receives $\boldsymbol{Y}$ according to \eqref{eq:ff-channel}. It decodes $\hat{\boldsymbol{S}} = \boldsymbol{a}$ if there is one and only one codeword $\boldsymbol{X}(\boldsymbol{a})$ that is jointly $\delta$-typical with the received signals, i.e.,
\begin{itemize}
\item $\Big(\boldsymbol{X}(\boldsymbol{a}),\boldsymbol{Y} \Big) \in \mathcal{A}^n_{[XY]\delta}$,\; and
\item $\Big(\boldsymbol{X}(\boldsymbol{b}),\boldsymbol{Y}\Big) \notin \mathcal{A}^n_{[XY]\delta}$,\; $\forall \boldsymbol{b} \in \mathcal{F}^k \setminus \{\boldsymbol{a}\}$.
\end{itemize}

Without loss of generality, let $\boldsymbol{S}=\boldsymbol{a}$ be the message sent.  The probability that the receiver makes an error in decoding is
\begin{subequations}
\begin{align}
P_\text{error} &= \Pr \{ \hat{\boldsymbol{S}} \neq \boldsymbol{a}\}\\
&= \Pr\Big\{ \Big(\boldsymbol{X}(\boldsymbol{a}),\boldsymbol{Y} \Big) \notin \mathcal{A}^n_{[XY]\delta} \text{ or } \Big(\boldsymbol{X}(\boldsymbol{b}),\boldsymbol{Y}\Big) \in \mathcal{A}^n_{[XY]\delta} \nonumber\\
&\quad\quad\quad\text{ for some } \boldsymbol{b} \neq \boldsymbol{a} \Big\} \\
&\leq \Pr\left\{ \Big(\boldsymbol{X}(\boldsymbol{a}),\boldsymbol{Y} \Big) \notin \mathcal{A}^n_{[XY]\delta} \right\} \nonumber\\
&\quad + \sum_{\boldsymbol{b} \neq \boldsymbol{a}} \Pr \left\{ \Big(\boldsymbol{X}(\boldsymbol{b}),\boldsymbol{Y}\Big) \in \mathcal{A}^n_{[XY]\delta} \right\}.
  \end{align}
\end{subequations}

From Lemma~\ref{lemma:jaep}, we have
\begin{equation}
\Pr\left\{ \Big(\boldsymbol{X}(\boldsymbol{a}),\boldsymbol{Y} \Big) \notin \mathcal{A}^n_{[XY]\delta} \right\}  < \delta.
\end{equation}

For any $\boldsymbol{b} \neq \boldsymbol{a}$, from Lemma~\ref{lemma:linear-codes-1} we know that $p\left(\boldsymbol{x}(\boldsymbol{b})\right) = \prod_{t=1}^np^{\text{u}}\left(x[t]\right)$, and from Lemma~\ref{lemma:linear-codes-2} we know that $\boldsymbol{x}(\boldsymbol{a})$ and $\boldsymbol{x}(\boldsymbol{b})$ are independent, and hence $p\left(\boldsymbol{x}(\boldsymbol{b}),\boldsymbol{y}\right) = \prod_{t=1}^np^{\text{u}}(x[t])p(y[t])$. So, from Lemma~\ref{lemma:jaep-2}, we have
\begin{equation}
\Pr \left\{ \Big(\boldsymbol{X}(\boldsymbol{b}),\boldsymbol{Y}\Big) \in \mathcal{A}^n_{[XY]\delta} \right\} \leq 2^{-n(I^{\text{u}}(X;Y) - 3\delta)},
\end{equation}
where $I^{\text{u}}(X;Y)$ is evaluated with $p(x,y) = p^{\text{u}}(x)p(y|x)$. Note that $p(y|x) = p(n)$.

This gives
\begin{subequations}
\begin{align}
P_\text{error} &\leq \delta + (|\mathcal{F}|^k-1)  2^{-n(I^{\text{u}}(X;Y) - 3\delta)}\\
& < \delta + 2^{n\left(\frac{k\log_2|\mathcal{F}|}{n}-[I^{\text{u}}(X;Y)-3\delta]\right)}.
\end{align}
\end{subequations}
Choosing a sufficiently large $n$ and a sufficiently small $\delta >0$, if
\begin{subequations}
\begin{align}
\frac{k\log_2|\mathcal{F}|}{n} &< I^{\text{u}}(X;Y)-3\delta\\
&=\log_2|\mathcal{F}| - H(N) - 3\delta,
\end{align}
\end{subequations}
then $P_\text{error}$ can be made as small as desired.

So, if $n$ is sufficiently large and if $\frac{k\log_2|\mathcal{F}|}{n} < \log_2|\mathcal{F}| - H(N)$, then the receiver can decode $\boldsymbol{S}$ with an arbitrarily small error probability.
\end{proof}

\begin{remark}\label{remark:all-rates}
Consider a message $w \in \{1,2,\dotsc, 2^{nR}\}$, and choose an integer $k$ such that
\begin{equation}
2^{nR} \leq |\mathcal{F}|^k \Leftrightarrow R \leq \frac{k\log_2|\mathcal{F}|}{n}. \label{eq:linear-rate}
\end{equation}
We can define an injective (one-to-one) function that maps each $w \in \{1,2,\dotsc, 2^{nR}\}$ to a unique $\boldsymbol{s} \in \mathcal{F}^k$, and send $\boldsymbol{s}$ using the linear code \eqref{eq:linear-codes-def-1} over $n$ uses of the channel \eqref{eq:ff-channel}. For any $R$ that satisfies
\begin{equation}
R < \log_2 |\mathcal{F}| - H(N),\label{linear-capacity-rate}
\end{equation}
we can always find sufficiently large $k$ and $n$, such that
\begin{equation}
R < \frac{k\log_2|\mathcal{F}|}{n} < \log_2 |\mathcal{F}| - H(N),
\end{equation}
meaning that the receiver can reliably decode $\boldsymbol{s}$, and it can then reverse the mapping from $\boldsymbol{s}$ to get the correct $w$. This means the rates in \eqref{linear-capacity-rate} are achievable using linear codes.
From \cite[pages 189-191]{coverthomas06}, the channel \eqref{eq:ff-channel} is \emph{symmetrical} and its capacity is $I(X;Y)$ evaluated with the uniform input distribution, i.e., $I^{\text{u}}(X;Y)=\log_2|\mathcal{F}| - H(N)$ bits/channel use. So, the random linear code defined in \eqref{eq:linear-codes-def-1} can be used to achieve the capacity of the channel \eqref{eq:ff-channel}.
\end{remark}

\section{Achievable Rate Region of Functional-Decode-Forward}\label{section:general-rate}

In this section, we extend the FDF scheme developed in~\cite{ongjohnsonkellett10cl} to MWRCs where the users are not constrained to transmitting at a common rate. Major differences are: (i) On the uplink, rate splitting is used, and (ii) On the downlink, joint source-channel decoding is used. Since rate splitting is used, we assume that the rates of all users, $R_i$, $\forall i \in \{1,2,\dotsc,L\}$, are rational numbers\footnote{Note that for the common-rate case, this is not required.}. The reason for this will become apparent later.

We consider $T$ message tuples.
Each user $i$, $i \in \{1,2,\dotsc,L\}$, sends $T$ messages of $nR_i$ bits each, meaning that each user can transmit at a different rate. Denote the $T$ messages of user $i$ by $(W_i[1], W_i[2], \dotsc, W_i[T])$ where $W_i[t] \in \{1,2,\dotsc,2^{nR_i}\}$ for all $t$.  Since we consider full data exchange, user $i$ needs to decode the messages sent by all the other users, i.e., $\Big\{W_{j}[t]: \forall j \in \{1,2,\dotsc,L\} \setminus \{i\}, \forall t \in \{1,2,\dotsc,T\} \Big\}$.

The message exchange among the users (via the relay) will be carried out in a total of $(T+1)$ \emph{blocks} of transmission. In the $t$-th block, for each $t \in \{1,2,\dotsc,T\}$, each user $i$ transmits (on the uplink) a codeword as a function of its $t$-th message $W_i[t]$. At the end of the $t$-th block, the relay decodes functions of its received signals in the $t$-th block. It then re-encodes these functions and transmits them (on the downlink) in the next block, i.e., the $(t+1)$-th block. At the end of the $(t+1)$-th block, each user $i$ then decodes the relay's transmission to obtain the $t$-th message of all other users, i.e., $\Big\{W_j[t]:j \in \{1,2,\dotsc,L\} \setminus \{i\} \Big\}$. So, for each pair of the $t$-th block on the uplink and the $(t+1)$-th block on the downlink, if each user can reliably decode the $t$-th message of all other users, then repeating the same coding scheme for all $t \in \{1,2,\dotsc,T\}$, at the end of $(T+1)$ blocks, all users will have reliably decoded the messages sent by all users transmitted in the first $T$ blocks. 

Let each block consist of $n$ channel uses, i.e., the entire transmission utilizes a total of $(T+1)n$ channel uses. Each user $i$ transmits a total of $TnR_i$ bits in this transmission period. If each user can reliably decode the messages of all other users, then the rate tuple $\left(\frac{TnR_1}{(T+1)n},\frac{TnR_2}{(T+1)n},\dotsc,\frac{TnR_L}{(T+1)n}\right)$ is achievable. For any $R_1$, $R_2$, $\dotsc$, $R_L$, and $n$, we can choose a sufficiently large $T$ such that the achievable rate tuple is arbitrarily close to $(R_1,R_2,\dotsc,R_L)$. In this section, we derive constraints on $R_1$, $R_2$, $\dotsc$, $R_L$ such that the rate tuple is achievable.

Since the encoding and decoding functions for all nodes are repeated in every block (different blocks for different message tuples), we focus on the first message tuple in Secs. \ref{uplink}, \ref{downlink}, and \ref{sec:decode}. The relevant channel uses are the first block on the uplink and the second block on the downlink. For simplicity, we denote $W_i[1]$ by $W_i$ in the these sections.

\subsection{On the Uplink}\label{uplink}

\noindent \underline{Message Splitting and Mapping:}

Recall that $R_i^{\text{c}} = \left(\sum_{j=1}^LR_j \right) - R_i$, $R_\text{min} = \min_{j \in \{1,2,\dotsc,L\}} R_j$ and $R_\text{min}^{\text{c}} = \left( \sum_{j=1}^L R_j \right) - R_\text{min}$.
For the uplink of the MWRC, we use the idea of FDF in~\cite{ongjohnsonkellett10cl} combined with rate splitting. For each user $i$, $i \in \{1,2,\dotsc,L\}$, we split its rate into
\begin{equation}
R_i = R_\text{min} + R_i',
\end{equation}
where $R_i' \geq 0$. So, each message $W_i$ can be split into
\begin{equation}
W_i = (A_i,B_i),
\end{equation}
where $A_i \in \{1,2,\dotsc,2^{nR_\text{min}}\}$ is a \emph{random} message of $nR_\text{min}$ bits in length and $B_i \in \{1,2,\dotsc,2^{nR_i'}\}$ is a \emph{random} message of $nR_i'$ bits in length\footnote{Since $R_\text{min}$ and $R_i'$, $\forall i$, are rational numbers, we can choose a sufficiently large $n$ such that $nR_\text{min}$ and $nR_i'$, $\forall i$, are integers.}. Let $D$, $0 \leq D  < L$, be the number of users whose message is strictly more than $nR_\text{min}$ bits. Let the set of these users be
\begin{equation}
\{d_1,d_2,\dotsc,d_D\} \triangleq \mathcal{D} \triangleq \{j: R_j' > 0\}.
\end{equation}
So, for all users $i \notin \mathcal{D}$, $W_i = A_i$, $B_i = \varnothing$, and $R_i'=0$.

On the downlink, we will invoke the result in Theorem~\ref{theorem:joint-source-channel-coding}, where the relay sends messages each consisting of $n_s$ i.i.d. random variables. To do this, we will further split each message into $n_s$ parts, i.e.,
\begin{align}
A_i &= (A_i^{(1)},A_i^{(2)}, \dotsc, A_i^{(n_s)}),\quad \forall i \in \{1,2,\dotsc, L\}\\
B_j &= (B_j^{(1)}, B_j^{(2)}, \dotsc, B_j^{(n_s)}), \quad \forall j \in \{d_1,d_2,\dotsc,d_D\},
\end{align}
where all $A_i^{(v)}$ are independently and uniformly distributed in $\{1,2,\dotsc,2^{nR_\text{min}/n_s}\}$, and all $B_j^{(v)}$ are independently and uniformly distributed in $\{1,2,\dotsc,2^{nR_j'/n_s}\}$. All these messages will be transmitted using linear codes in $\mathcal{F}$ defined in \eqref{eq:linear-codes-def-1}. To do this, we define an injective function that maps each $\alpha \in \{1,2,\dotsc,2^{nR_\text{min}/n_s}\}$ to a unique length-$k_\text{A}$ finite field vector $\boldsymbol{s}(\alpha) \in \mathcal{F}^{k_\text{A}}$. This means the vector length $k_\text{A}$ must be chosen such that
\begin{subequations}
\begin{align}
2^{nR_\text{min}/n_s} &\leq |\mathcal{F}|^{k_\text{A}}\\
\frac{k_\text{A}n_s\log_2|\mathcal{F}|}{n} &\geq R_\text{min}. \label{eq:mapping-1}
\end{align}
\end{subequations}
This guarantees that a user can always reverse the function to get the correct $A_i^{(v)}$ from $\boldsymbol{S}(A_i^{(v)})$. Similarly, for each $j \in \mathcal{D}$, we define an injective function that maps each $\beta_j \in \{1,2,\dotsc,2^{nR_j'/n_s}\}$ to a unique length-$k_{\text{B},j}$ finite field vector $\boldsymbol{s}(b_j) \in \mathcal{F}^{k_{\text{B},j}}$. So, $k_{\text{B},j}$ must be chosen such that
\begin{equation}\label{eq:mapping-2}
\frac{k_{\text{B},j} n_s \log_2|\mathcal{F}|}{n} \geq R_{j}'.
\end{equation}
The length of the vector $\boldsymbol{s}(\gamma)$ and the corresponding mapping is clear from its argument $\gamma \in \{\alpha, \beta_j\}$.

\begin{figure*}[t]
  \centering
  \includegraphics[width=13.5cm]{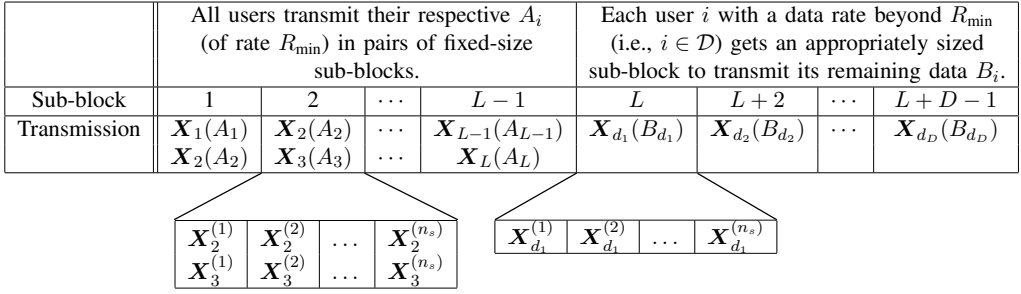}
  \caption{Uplink transmission}
\label{fig:general-rate-uplink}
  \end{figure*}

\noindent \underline{Transmission:}

The block of $n$ uplink channel uses are split into $(L+D-1)$ sub-blocks. Each of the $l$-th sub-blocks for $1 \leq l \leq L-1$ consists of $\frac{nR_\text{min}}{R_\text{min}^{\text{c}}}$ channel uses\footnotemark[4]. Each of the $l$-th sub-blocks for $L \leq l \leq L+D-1$ consists of $\frac{nR_{d_{l-L+1}}'}{R_\text{min}^{\text{c}}}$ channel uses\footnotemark[4]. Note that if we sum the number of channel uses in all sub-blocks, we get
\begin{equation}
(L-1)\frac{nR_\text{min}}{R_\text{min}^{\text{c}}} + \sum_{d \in \mathcal{D}}\frac{nR_d'}{R_\text{min}^{\text{c}}} = n\frac{\sum_{j=1}^L (R_\text{min} + R_j')- R_\text{min}}{R_\text{min}^{\text{c}}} = n.
\end{equation}

The first $(L-1)$ sub-blocks (of equal length) are used to send $\left\{A_i: i \in \{1,2,\dotsc,L\}\right\}$. 
The next $D$ sub-blocks (of possibly different length) are used to send $\{B_j: j \in \mathcal{D}\}$. 

In the $l$-th sub-block for $l \in \{1,2,\dotsc, L-1\}$, only two users (more specifically, users $l$ and $(l+1)$) transmit, and the rest of the users \emph{do not transmit} (which is defined as transmitting the additive identity $\mathfrak{0}$). Define the transmission of user $i$ in the sub-block as
\begin{equation}
\boldsymbol{X}_i = (\boldsymbol{X}_i^{(1)}, \boldsymbol{X}_i^{(2)}, \dotsc, \boldsymbol{X}_i^{(n_s)}). \label{eq:further-split}
\end{equation}
The two \emph{active} users transmit using linear codes in $\mathcal{F}$ of the form defined in \eqref{eq:linear-codes-def-1}, i.e.,
\begin{equation}
\boldsymbol{X}_i^{(v)} = \begin{cases}
(\boldsymbol{S}(A_i^{(v)}) \odot \mathbb{G}_\text{A}) \oplus \boldsymbol{q}_{\text{A},i}, &\text{if } i=l \text{ or } l+1\\
\boldsymbol{\mathfrak{0}}, &\text{otherwise},
\end{cases}
\end{equation}
for all $v \in \{1,2,\dotsc, n_s\}$,
where each $\boldsymbol{S}(A_i^{(v)})$ is a row vector of length $k_\text{A}$, $\mathbb{G}_\text{A}$ is a fixed $k_\text{A} \times \frac{nR_\text{min}}{n_sR_\text{min}^{\text{c}}}$ matrix\footnote{For any (possibly large) $n_s$, we choose a much larger $n$ such that $\frac{n}{n_s}$ is sufficiently large, so that $\frac{nR_\text{min}}{n_sR_\text{min}^{\text{c}}}$ and all $\frac{nR_{d_m}'}{n_sR_\text{min}^{\text{c}}}$ are integers.}, each $\boldsymbol{X}_i^{(v)}$ and $\boldsymbol{q}_{\text{A},i}$ is a row vector of length $\frac{nR_\text{min}}{n_sR_\text{min}^{\text{c}}}$, and $\boldsymbol{\mathfrak{0}}$ is the all-zero row vector. Each element in the vectors/matrix is over $\mathcal{F}$. 

For the next $D$ sub-blocks, only users in $\mathcal{D}$ (those with an ``extra'' message $B_i$) transmit. We use the same notation in \eqref{eq:further-split} for the transmitted symbols. More specifically, in the $(L-1+m)$-th sub-block for $m \in \{1,2,\dotsc, D\}$, only one user, $d_m \in \mathcal{D}$, transmits, and does so using a linear code of the form defined in \eqref{eq:linear-codes-def-1}, i.e., 
\begin{align}
\boldsymbol{X}_i^{(v)} = \begin{cases}
(\boldsymbol{S}(B_i^{(v)}) \odot \mathbb{G}_{\text{B},i}) \oplus \boldsymbol{q}_{\text{B},i}, &\text{if } i = d_m\\
\boldsymbol{\mathfrak{0}}, &\text{otherwise},
\end{cases}
\end{align}
for all $v \in \{1,2,\dotsc,n_s\}$,
where $\boldsymbol{S}(B_{d_m}^{(v)})$ is a row vector of length $k_{\text{B},d_m}$, $\mathbb{G}_{\text{B},d_m}$ is a fixed $k_{\text{B},d_m} \times \frac{nR_{d_m}'}{n_sR_\text{min}^{\text{c}}}$ matrix\footnotemark[5], and each $\boldsymbol{X}_{d_m}^{(v)}$ and  $\boldsymbol{q}_{\text{B},d_m}$ is a fixed row vector of length $\frac{nR_{d_m}'}{n_sR_\text{min}^{\text{c}}}$.

Each element in $\mathbb{G}_\text{A}$, $\mathbb{G}_{\text{B},d_m}$, $\boldsymbol{q}_{\text{A},i}$, and $\boldsymbol{q}_{\text{B},d_m}$  is independently and uniformly chosen over $\mathcal{F}$, is fixed for all transmissions, and is made known to the relay. The transmission scheme above is summarized in Fig.~\ref{fig:general-rate-uplink}.

\noindent \underline{Decoding:}

In the $l$-th sub-block for $l \in \{1,2,\dotsc, L-1\}$, the relay receives $\boldsymbol{Y}_0 = (\boldsymbol{Y}_0^{(1)}, \boldsymbol{Y}_0^{(2)},\dotsc, \boldsymbol{Y}_0^{(n_s)})$, where $\boldsymbol{Y}_0^{(v)} = \boldsymbol{X}_{l,l+1}^{(v)} \oplus \boldsymbol{N}_0^{(v)}$ and
\begin{align}
\boldsymbol{X}_{l,l+1}^{(v)} &= \Big([(h_{l,0} \odot \boldsymbol{S}(A_l^{(v)})) \oplus (h_{l+1,0} \odot \boldsymbol{S}(A_{l+1}^{(v)}))] \odot \mathbb{G}_\text{A} \Big) \nonumber\\
&\quad \oplus (\boldsymbol{q}_{\text{A},l} \oplus \boldsymbol{q}_{\text{A},l+1}),
\end{align}
which is also a linear codeword of the form \eqref{eq:linear-codes-def-1}, where the ``message'' is
\begin{equation}
\boldsymbol{S}(A_{l,l+1}^{(v)}) \triangleq (h_{l,0} \odot \boldsymbol{S}(A_l^{(v)})) \oplus (h_{l+1,0} \odot \boldsymbol{S}(A_{l+1}^{(v)})) \in \mathcal{F}^{k_\text{A}}.
\end{equation}
From Theorem~\ref{theorem:ff-linear-code}, if $\frac{nR_\text{min}}{n_sR_\text{min}^{\text{c}}}$ is sufficiently large and if
\begin{equation}
\frac{k_\text{A} \log_2|\mathcal{F}|}{\frac{nR_\text{min}}{n_sR_\text{min}^{\text{c}}}} < \log_2|\mathcal{F}| - H(N_0),\label{eq:uplink-1}
\end{equation}
then the relay can reliably decode $\boldsymbol{S}(A_{l,l+1}^{(v)})$, for all $v \in \{1,2,\dotsc,n_s\}$.

In the $(m+L-1)$-th sub-block for  $m \in \{1,2,\dotsc,D\}$, only one user $d_m$ transmits at any time. The relay  scales each of its received signals by $h_{d_m,0}^{-1}$ (the \emph{multiplicative inverse} of $h_{d_m,0}$) to get
\begin{equation}
\tilde{Y}_0 = h_{d_m,0}^{-1} \odot Y_i = X_{d_m} \oplus \tilde{N}_0,
\end{equation}
where $\tilde{N}_0 = h_{d_m,0}^{-1} \odot N_0$. Note that $H(\tilde{N}_0)=H(N_0)$ as, for any fixed $h_{d_m,0}^{-1} \neq \mathfrak{0}$, there is a bijective mapping between the two random variables $\left(h_{d_m,0}^{-1} \odot N_0\right)$ and $N_0$. Applying Theorem~\ref{theorem:ff-linear-code}, if $\frac{nR_{d_m}'}{n_sR_\text{min}^{\text{c}}}$ is sufficiently large and if
\begin{equation}
\frac{k_{\text{B},d_m} \log_2|\mathcal{F}|}{\frac{nR_{d_m}'}{n_sR_\text{min}^{\text{c}}}} < \log_2|\mathcal{F}| - H(\tilde{N}_0) = \log_2|\mathcal{F}| - H(N_0), \label{eq:uplink-2}
\end{equation}
then the relay can reliably decode $\boldsymbol{S}(B_{d_m}^{(v)})$ from $\tilde{\boldsymbol{Y}}_0^{(v)} = \boldsymbol{X}_{d_m}^{(v)} + \tilde{\boldsymbol{N}}_0^{(v)}$, for all $v \in \{1,2,\dotsc,n_s\}$.

Define
\begin{align}
\boldsymbol{U}^{(v)} &\triangleq \Big(\boldsymbol{S}(A_{1,2}^{(v)}), \boldsymbol{S}(A_{2,3}^{(v)}), \dotsc, \boldsymbol{S}(A_{L-1,L}^{(v)}), \nonumber\\
&\quad\quad \boldsymbol{S}(B_{d_1}^{(v)}),\boldsymbol{S}(B_{d_2}^{(v)}),\dotsc,  \boldsymbol{S}(B_{d_D}^{(v)})\Big),
\end{align}
and
\begin{equation}
\mathbb{U} \triangleq (\boldsymbol{U}^{(1)},\boldsymbol{U}^{(2)},\dotsc,\boldsymbol{U}^{(n_s)}).
\end{equation}
On the uplink, if
\begin{equation}
R_\text{min}^{\text{c}} < \log_2|\mathcal{F}| - H(N_0), \label{eq:uplink-3}
\end{equation}
we can always find sufficiently large $\frac{n}{n_s}$, $k_\text{A}$, and $\{k_{\text{B},d_m}\}_{d_m \in \mathcal{D}}$, such that
\begin{align}
R_\text{min}^{\text{c}} &\leq R_\text{min}^{\text{c}}\frac{k_\text{A} n_s \log_2|\mathcal{F}|}{nR_\text{min}} < \log_2|\mathcal{F}| - H(N_0) \label{eq:arbitrari-close} \\
R_\text{min}^{\text{c}} &\leq R_\text{min}^{\text{c}}\frac{k_{\text{B},d_m} n_s \log_2|\mathcal{F}|}{nR_{d_m}'}  \nonumber\\
&< \log_2|\mathcal{F}| - H(N_0),\quad \forall d_m \in \mathcal{D},
\end{align}
meaning that \eqref{eq:mapping-1}, \eqref{eq:uplink-1} and \eqref{eq:mapping-2}, \eqref{eq:uplink-2} can be satisfied in their respective sub-blocks. So, if \eqref{eq:uplink-3} is satisfied and if $\frac{n}{n_s}$ is sufficiently large, the relay can reliably decode $\mathbb{U}$.

Eqns.~\eqref{eq:uplink-3} and \eqref{eq:arbitrari-close} also mean that $\frac{k_\text{A} n_s\log_2|\mathcal{F}|}{n}$ can be chosen arbitrarily close to $R_\text{min}$, i.e.,
\begin{equation}
\frac{k_\text{A} n_s\log_2|\mathcal{F}|}{n} = R_\text{min} + \eta, \label{eq:eta}
\end{equation}
where $\eta > 0$ can be chosen arbitrarily small.

\subsection{On the Downlink}\label{downlink}

Now, assume that the relay decodes $\mathbb{U}$ in the first block of $n$ uplink uses, it broadcasts this information in the second block of $n$ downlink uses. For decoding on the downlink, each user $i$, $i \in \{1,2,\dotsc,L\}$, scales each of its received signals by $h_{0,i}^{-1}$ to get
\begin{equation}
\tilde{Y}_i = h_{0,i}^{-1} \odot Y_i = X_0 \oplus \tilde{N}_i,
\end{equation}
where $\tilde{N}_i = h_{0,i}^{-1} \odot N_i$, and $H(\tilde{N}_i)=H(N_i)$.

Note that each $\boldsymbol{U}^{(v)}$ is i.i.d., for all $v \in \{1,2,\dotsc,n_s\}$, so are $\boldsymbol{S}(A_{i,i+1}^{(v)})$ for all $v$, and $\boldsymbol{S}(B_i^{(v)})$ for all $v$. We use $\boldsymbol{U}$, $\boldsymbol{S}_{i,i+1}$, and $\boldsymbol{S}_i$ to denote the respective generic random variables. Thus, we have $\boldsymbol{U} = ( \boldsymbol{S}_{1,2}, \boldsymbol{S}_{2,3}, \dotsc, \boldsymbol{S}_{L,L-1}, \boldsymbol{S}_{d_1}, \boldsymbol{S}_{d_2}, \dotsc, \boldsymbol{S}_{d_D} )$.

With this, we can re-cast the downlink as a broadcast channel in which the relay broadcasts a message $\mathbb{U}=[\boldsymbol{U}^{(v)}]_{\forall v}$ to all the users, where each user $i \in \mathcal{D}$ knows $[\boldsymbol{S}(B_i^{(v)})]_{\forall v}$ (which is correlated with the message $\mathbb{U}$) \emph{a priori}. So, each user $i \in \mathcal{D}$ can use its \emph{side information} $[\boldsymbol{S}(B_i^{(v)})]_{\forall v}$ to decode $\mathbb{U}$ from its scaled received signals $\tilde{\boldsymbol{Y}}_i$ during channel decoding (hence joint source-channel decoding). Note that all users do not need to use their respective $A_i$ as side information for decoding $\mathbb{U}$ (see Remark~\ref{remark:U-and-A}).  From Theorem~\ref{theorem:joint-source-channel-coding}, all users can reliably decode $\mathbb{U}$ if $n_s$ and $n$ are sufficiently large and if
\begin{align}
n_sH(\boldsymbol{U}|\boldsymbol{S}_i) < nI(X_0;\tilde{Y}_i),\quad &\forall i \in \mathcal{D} \label{eq:down-mwrc-1}\\
n_sH(\boldsymbol{U}) < nI(X_0;\tilde{Y}_i),\quad &\forall i \notin \mathcal{D}, \label{eq:down-mwrc-2}
\end{align}
for some $p(x_0)$. Note that $\boldsymbol{S}(B_i^{(v)}) = \varnothing$ if $i \notin \mathcal{D}$. Choosing the uniform distribution for $X_0$, $I(X_0;\tilde{Y}_i)  = \log_2|\mathcal{F}| - H(\tilde{N}_i) = \log_2|\mathcal{F}| - H(N_i)$, for all $i \in \{1,2,\dotsc,L\}$.

Since the mapping from $B_i^{(v)}$ (which is uniformly distributed in $\{1,2,\dotsc, 2^{nR_i'/n_s}\}$) to $\boldsymbol{S}(B_i^{(v)})$ is injective, we have, for all $i \in \mathcal{D}$,
\begin{equation}
H(\boldsymbol{S}_i) = \frac{nR_i'}{n_s}.
\end{equation}
Since $\boldsymbol{S}_{i,i+1} \in \mathcal{F}^{k_\text{A}}$, we have
\begin{equation}
H(\boldsymbol{S}_{i,i+1}) \leq k_\text{A}\log_2|\mathcal{F}|, \label{eq:U-ka}
\end{equation}
with equality if and only if $\boldsymbol{S}_{i,i+1}$ is uniformly distributed in $\mathcal{F}^{k_\text{A}}$. Note that each $A_i^{(v)}$, $\forall i$, being uniformly distributed does not imply that $\boldsymbol{S}(A_{i,i+1}^{(v)})$ is uniformly distributed.

This gives
\begin{subequations}
\begin{align}
&H(\boldsymbol{U}) \nonumber\\
 &= \sum_{i=1}^{L-1} H\Big(\boldsymbol{S}_{i,i+1} \Big| \big\{\boldsymbol{S}_{j,j+1}: \text{for all } j < i \text{ and } j \geq 1\big\}\Big) \nonumber\\
&\quad + \sum_{k=1}^D H \Big(\boldsymbol{S}_{d_k} \Big| \big\{\boldsymbol{S}_{d_\ell}: \text{for all } \ell < k \text{ and } \ell \leq 1 \big\}, \nonumber\\
&\quad\quad\quad\quad\quad\quad\quad\;\, \big\{\boldsymbol{S}_{m,m+1}: 1 \leq m \leq L-1 \big\} \Big) \label{eq:chain-rule-1}\\
&\leq \left(\sum_{i=1}^{L-1} H(\boldsymbol{S}_{i,i+1}) + \sum_{d \in \mathcal{D}}H(\boldsymbol{S}_d) \right) \label{eq:conditioning-1}\\
&\leq (L-1)k_\text{A}\log_2|\mathcal{F}| + \sum_{d \in \mathcal{D}} \frac{nR_d'}{n_s}\\
&= (L-1)\frac{n}{n_s}(R_\text{min} + \eta) + \frac{n}{n_s}\sum_{d \in \mathcal{D}} R_d' \label{eq:eta-greater-than-zero}\\
& = \frac{n}{n_s} \left((L-1)R_\text{min} + \sum_{d \in \mathcal{D}} R_d' + (L-1)\eta \right)\\
& = \frac{n}{n_s} \left(R_\text{min}^{\text{c}} + \zeta \right), \label{eq:zeta-1}
\end{align}
\end{subequations}
where $\eta$ is defined in \eqref{eq:eta}, and $\zeta = (L-1)\eta>0$ can be chosen arbitrarily small.
Here, \eqref{eq:chain-rule-1} follows from the chain rule, and \eqref{eq:conditioning-1} is because conditioning can only reduce entropy. 

It follows that for all $i \in \mathcal{D}$,
\begin{subequations}
\begin{align}
H(\boldsymbol{U}|\boldsymbol{S}_i) &= H(\boldsymbol{U}) + H(\boldsymbol{S}_i|\boldsymbol{U}) - H(\boldsymbol{S}_i)\\
&=  H(\boldsymbol{U}) - H(\boldsymbol{S}_i) \label{eq:u-contain-s}\\
& \leq \frac{n}{n_s} \left(R_\text{min}^{\text{c}} + \zeta - R_i'\right)\\
& = \frac{n}{n_s} \left(\left(\sum_{j=1}^LR_j \right) - R_\text{min} - R_i' + \zeta\right)\\
& =\frac{n}{n_s} \left(R_i^{\text{c}} + \zeta\right), \label{eq:zeta-2}
\end{align}
\end{subequations}
where $\zeta > 0$ can be chosen arbitrarily small. Here, \eqref{eq:u-contain-s} is because $ H(\boldsymbol{S}_i|\boldsymbol{U})=0$.

Note that for all $i \notin \mathcal{D}$, $R_i'=0$, meaning $R_i = R_\text{min}$, and hence $R_i^{\text{c}} = R_\text{min}^{\text{c}}$.
Now, for all $i \in \{1,2,\dotsc,L\}$, if
\vspace{-0.7ex}
\begin{equation}
R_i^{\text{c}} < \log_2|\mathcal{F}| - H(N_i), \label{eq:capacity-mwrc-2}
\end{equation}
which is equivalent to 
\begin{equation}
R_i^{\text{c}} + \psi = \log_2|\mathcal{F}| - H(N_i), \quad \text{for some } \psi > 0,
\end{equation}
we can then choose $\zeta = \frac{\psi}{2}$ for \eqref{eq:zeta-1} and \eqref{eq:zeta-2} so that \eqref{eq:down-mwrc-1} and \eqref{eq:down-mwrc-2} can both be satisfied, i.e., all users can reliably decode $\mathbb{U}$ with sufficiently large $n_s$ and $n$.

Note that on the downlink, linear codes are not required.

\begin{remark}\label{remark:U-and-A}
Consider the two-user case (i.e., $L=2$) where $R_1 = R_2 = R_\text{min}$. So, the two messages are $W_1=A_1$ and $W_2=A_2$. Ideally, we choose $k_\text{A}$ such that $nR_\text{min}/n_s  \overset{\text{c.f. \eqref{eq:eta}}}{\approx}  k_\text{A} \log_2|\mathcal{F}|  \overset{\text{c.f. \eqref{eq:U-ka}}}{\approx}  H(\boldsymbol{S}(A_{1,2}^{(v)}))$. Since, $\boldsymbol{U}^{(v)} = \boldsymbol{S}(A_{1,2}^{(v)})$, we have $H(\boldsymbol{U}^{(v)}) = H(\boldsymbol{S}(A_{1,2}^{(v)})) \approx k_\text{A} \log_2|\mathcal{F}|$. Since $A_1^{(v)}$ and $A_2^{(v)}$ are uniformly distributed in $\{1,2,\dotsc,2^{nR_\text{min}/n_s}\}$, we have $H(A_1^{(v)}) = H(A_2^{(v)}) = nR_\text{min}/n_s$. Given $A_1^{(v)}$, the only uncertainty left in $\boldsymbol{U}^{(v)}$ is that of $A_2^{(v)}$. This means $H(\boldsymbol{U}^{(v)}|A_1^{(v)}) = H(A_2^{(v)}) = nR_\text{min}/n_s \approx k_\text{A} \log_2|\mathcal{F}| \approx H(\boldsymbol{U}^{(v)})$. Similarly, we can show that $H(\boldsymbol{U}^{(v)}|A_2^{(v)}) \approx H(\boldsymbol{U}^{(v)})$. So, each message, $A_1^{(v)}$ or $A_2^{(v)}$, individually conveys very little information about $\boldsymbol{U}^{(v)}$. This explains why we do not lose optimality by not using $A_i$ as side information when each user decodes $\mathbb{U}$ on the downlink.
\end{remark}

\subsection{Decoding of Other Users' Messages}\label{sec:decode}

Assume that every user $i$, for all $i \in \{1,2,\dotsc,L\}$, correctly decodes $\mathbb{U}$, i.e., $\boldsymbol{U}^{(v)} \triangleq \Big(\boldsymbol{S}(A_{1,2}^{(v)}), \boldsymbol{S}(A_{2,3}^{(v)}),$ $\dotsc,
\boldsymbol{S}(A_{L-1,L}^{(v)}),\boldsymbol{S}(B_{d_1}^{(v)}),\boldsymbol{S}(B_{d_2}^{(v)}),\dotsc,  \boldsymbol{S}(B_{d_D}^{(v)})\Big)$ for all $v \in \{1,2,\dotsc,n_s\}$, sent by the relay. Since \eqref{eq:mapping-2} is true, user $i$ can correctly decode $B_{j}^{(v)}$ from $\boldsymbol{S}(B_{j}^{(v)})$, for all $j \in \mathcal{D}$. Recall that $B_k^{(v)} = \varnothing$, for all $k \notin \mathcal{D}$.

Then user $i$ performs the following:
\begin{subequations}
\begin{align}
\boldsymbol{S}(A_{i+1}^{(v)}) &= ( h_{i+1,0}^{-1} \odot \boldsymbol{S}(A_{i,i+1}^{(v)}))\nonumber\\
&\quad \oplus -(h_{i+1,0}^{-1} \odot h_{i,0} \odot \boldsymbol{S}(A_i^{(v)})) \label{eq:chain-decoding-start} \\
\boldsymbol{S}(A_{i+2}^{(v)}) &= ( h_{i+2,0}^{-1} \odot \boldsymbol{S}(A_{i+1,i+2}^{(v)})) \nonumber\\
&\quad  \oplus -(h_{i+2,0}^{-1}  \odot h_{i+1,0} \odot \boldsymbol{S}(A_{i+1}^{(v)}) )\\
& \vdots \nonumber\\
\boldsymbol{S}(A_{L}^{(v)}) &= ( h_{L,0}^{-1} \odot \boldsymbol{S}(A_{L-1,L}^{(v)})) \nonumber\\
&\quad  \oplus -(h_{L,0}^{-1} \odot h_{L-1,0} \odot \boldsymbol{S}(A_{L-1}^{(v)}) )\\
\boldsymbol{S}(A_{i-1}^{(v)}) &= ( h_{i-1,0}^{-1} \odot \boldsymbol{S}(A_{i-1,i}^{(v)}) \nonumber\\
&\quad  \oplus -(h_{i-1,0}^{-1} \odot h_{i,0} \odot \boldsymbol{S}(A_{i}^{(v)} ) )\\
\boldsymbol{S}(A_{i-2}^{(v)}) &= ( h_{i-2,0}^{-1} \odot \boldsymbol{S}(A_{i-2,i-1}^{(v)}) \nonumber\\
&\quad  \oplus -(h_{i-2,0}^{-1} \odot h_{i-1,0} \odot \boldsymbol{S}(A_{i-1}^{(v)}) )\\
& \vdots \nonumber\\
\boldsymbol{S}(A_{1}^{(v)}) &= ( h_{1,0}^{-1} \odot \boldsymbol{S}(A_{1,2}^{(v)}) \oplus -(h_{1,0}^{-1} \odot h_{2,0} \odot \boldsymbol{S}(A_{2}^{(v)} ) ), \label{eq:chain-decoding-stop}
\end{align}
\end{subequations}
to get $(\boldsymbol{S}(A_1^{(v)}), \boldsymbol{S}(A_{2}^{(v)}), \dotsc, \boldsymbol{S}(A_{i-1}^{(v)}), \boldsymbol{S}(A_{i+1}^{(v)}),\dotsc,$ $\boldsymbol{S}(A_{L}^{(v)}))$.
Since \eqref{eq:mapping-1} is true, user $i$ can correctly decode $A_j^{(v)}$ from $\boldsymbol{S}(A_j^{(v)})$, for all $j \in \{1,2,\dotsc,L\} \setminus \{i\}$. Repeating that for all $v \in \{1,2,\dotsc,n_s\}$, user $i$ then obtains all other users' messages, i.e., $\Big\{W_j=(A_j,B_j): j \in \{1,2,\dotsc,L\} \setminus \{i\} \Big\}$.

\subsection{Probability of Error} \label{sec:pe}

In the above analyses, we focused on the first message tuple. Now, we consider all $T$ message tuples. On the uplink, let the decoding error at the relay in the $v$-th fraction of the $l$-th sub-block of the $t$-th message tuple be $P_\text{e}(0,t,l,v)$, for $t \in \{1,2,\dotsc,T\}$, $l \in \{1,2,\dotsc,L+D-1\}$, and $v \in \{1,2,\dotsc,n_s\}$. On the downlink, let the decoding error at user $i$ (of the message $\mathbb{U}$ sent by the relay) of the $t$-th message tuple be $P_\text{e}(i,t)$, for $i \in \{1,\dotsc,L\}$ and $t \in \{1,2,\dotsc,T\}$. 

For the $t$-th message tuple, from Section~\ref{uplink}, if $\frac{n}{n_s}$ is sufficiently large and if \eqref{eq:uplink-3} is satisfied, then $P_\text{e}(0,t,l,v) < \epsilon_1$ for any $\epsilon_1 > 0$, for all $l$ and $v$, meaning that the relay can reliably decode $\mathbb{U}$. If the relay correctly decodes $\mathbb{U}$ (of the $t$-th message tuple) and transmits it on the downlink, from Section~\ref{downlink}, with $n_s$ and $n$ sufficiently large and \eqref{eq:capacity-mwrc-2} satisfied, all users can reliably decode $\mathbb{U}$, i.e., $P_\text{e}(i,t) < \epsilon_2$ for any $\epsilon_2 >0$, for all $i \in \{1,2,\dotsc,L\}$.

Note that $P_\text{e}(i,t)$ for the users, i.e., $i \neq 0$, are found conditioned on the event that the relay has correctly decoded $\mathbb{U}$ (of the $t$-th message tuple in the previous block of transmission). When we calculate the \emph{end-to-end} error probability, $P_\text{e}$, in the remaining of the section, we will show that the event that the relay wrongly decodes (or correctly decodes parts of) $\mathbb{U}$ can be made arbitrarily small  (i.e., we do not assume that the relay correctly decodes $\mathbb{U}$). Combining this with the fact that the probability that some users wrongly decode (or correctly decode parts of) $\mathbb{U}$ given the relay has correctly decoded $\mathbb{U}$ can also be made arbitrarily small, we can make $P_\text{e}$ as small as desired. If the relay makes a decoding error, the error propagates onto the downlink to the users. But we can make the probability of this event arbitrarily small.

Now, if \eqref{eq:uplink-3} is satisfied, we have
\begin{subequations}
\begin{align}
&\Pr\{ \text{Relay makes some decoding error(s)} \}\nonumber\\
&\leq \sum_{t=1}^T \sum_{l=1}^{L+D-1} \sum_{v=1}^{n_s} \Pr \Big\{ \text{Relay wrongly decodes } \boldsymbol{S}(A_{l,l+1}^{(v)}) \text{ or } \nonumber\\
&\quad\quad\quad\quad\quad\quad\quad\quad\quad\,\,\,  \boldsymbol{S}(B_{d_{l-L+1}}^{(v)}) \text{in the $l$-th sub-block for }\nonumber\\
&\quad\quad\quad\quad\quad\quad\quad\quad\quad\,\,\, \text{the $t$-th message tuple} \Big\}\\
& = \sum_{t=1}^T \sum_{l=1}^{L+D-1} \sum_{v=1}^{n_s} P_\text{e}(0,t,l,v)\\
& \leq (L+D-1)Tn_s\epsilon_1,
\end{align}
\end{subequations}
and so
\begin{equation}
\Pr\{ \text{Relay makes no error} \} \geq 1 - (L+D-1)Tn_s\epsilon_1.
\end{equation}

Conditioned on the event that the relay makes no decoding error,  if \eqref{eq:capacity-mwrc-2} is satisfied, we have
\begin{subequations}
\begin{align}
&\Pr\Big\{\text{Some user(s) makes some decoding error(s) }\nonumber\\
&\quad\quad \Big| \text{ Relay makes no error} \Big\}\nonumber\\
&\leq \sum_{i=1}^L \Pr\Big\{\text{User $i$ makes some decoding error(s) } \nonumber\\
&\quad\quad\quad\quad\;\;\;  \Big| \text{ Relay makes no error} \Big\}\\
&\leq \sum_{i=1}^L \sum_{t=1}^T P_\text{e}(i,t)\\
&\leq LT\epsilon_2,
\end{align}
\end{subequations}
and so
\begin{multline}
\Pr\big\{\text{No user makes any decoding error }\\ \big| \text{ Relay makes no error} \big\} \geq 1- LT\epsilon_2.
\end{multline}

This gives
\begin{multline}
\Pr\{ \text{No user makes any decoding error}\}\\ > [1 - (L+D-1)Tn_s\epsilon_1][1- LT\epsilon_2],
\end{multline}
and
\begin{subequations}
\begin{align}
P_\text{e} &\triangleq \Pr\{ \text{Some user(s) makes some error(s)} \}\\
&< 1 - [1 - (L+D-1)Tn_s\epsilon_1][1- LT\epsilon_2] \label{eq:error-probability}\\
&< (L+D-1)Tn_s\epsilon_1 + LT\epsilon_2 - (L+D-1)LT^2n_s\epsilon_1\epsilon_2, \label{eq:small}
\end{align}
\end{subequations}
where $\epsilon_1 \rightarrow 0$ as $\frac{n}{n_s} \rightarrow \infty$, and $\epsilon_2 \rightarrow 0$ as $n_s,n \rightarrow \infty$. The RHS of \eqref{eq:small} can be made arbitrarily small for any $L$, $T$, $D$ (note that $D < L$),  by choosing a sufficiently large $n_s$ and much larger $n$, such that $\frac{n}{n_s}$ is also sufficiently large, making $P_\text{e}$ arbitrarily small.

\subsection{The Capacity Region of the MWRC over a Finite Field}

The preceding analysis means that all rate tuples $(R_1,R_2,\dotsc,R_L)$ satisfying \eqref{eq:uplink-3} and \eqref{eq:capacity-mwrc-2} are achievable. Comparing this achievable region with the capacity upper bound in Theorem~\ref{theorem:general-upper-bound}, we have the following capacity theorem.

\begin{theorem}\label{theorem:mwrc-capacity}
Consider the $L$-user MWRC over a finite field $\mathcal{F}$. The capacity region is the set of all non-negative rate tuples $(R_1,R_2,\dotsc,R_L)$ satisfying
\begin{align}
R_\text{min}^{\text{c}} &\leq \log_2|\mathcal{F}| - H(N_0) \label{eq:mwrc-capacity-1}\\
R_i^{\text{c}} &\leq \log_2|\mathcal{F}| - H(N_i), \quad \forall i \in \{1,2,\dotsc,L\}. \label{eq:mwrc-capacity-2}
\end{align}
\end{theorem}

\begin{remark}
Note that in the FDF coding strategy proposed above, each user's transmitted signals only depend on its message and do not depend on its received signals, i.e., $X_i[t] = f_{i,t}(W_i)$, $\forall i,t$. Since this is sufficient to achieve the capacity region, the capacity region remains the same even if we consider the \emph{restricted} MWRC where the users' transmitted signals can only depend on their respective messages and cannot depend on their received signals. This means utilizing feedback does not increase the capacity region of MWRCs over finite fields.
\end{remark}

\begin{remark}\label{remark:capacity}
The capacity region in Theorem~\ref{theorem:mwrc-capacity} is equivalent to the set of all rate tuples $(R_1,R_2,\dotsc,R_L)$ satisfying
\begin{equation}
R_i^{\text{c}} \leq \log_2|\mathcal{F}| - \max \{H(N_0), H(N_i)\},\label{eq:cap}
\end{equation}
for all $i \in \{1,2,\dotsc,L\}$.
\end{remark}

Now, we show that the capacity region in Remark~\ref{remark:capacity}, denoted by $\mathcal{R}$, is convex and hence the convex hull operation is not required. Let two rate tuples be $(R_1^{(1)}, R_2^{(1)}, \dotsc, R_L^{(1)}), (R_1^{(2)}, R_2^{(2)}, \dotsc, R_L^{(2)}) \in \mathcal{R}$. For any $\alpha \in [0,1]$, define $(R_1^{(3)}, R_2^{(3)}, \dotsc, R_L^{(3)})$ such that $R_i^{(3)} = \alpha R_i^{(1)} + (1-\alpha) R_i^{(2)}$, $\forall i$. For this rate tuple, and for all $i \in \{1,2,\dotsc,L\}$, we have
\begin{subequations}
\begin{align}
R_i^{(3)c} &\triangleq \sum_{j=1}^L R_j^{(3)} - R_i^{(3)}\\
&= \sum_{j=1}^L (\alpha R_j^{(1)} + (1-\alpha) R_j^{(2)}) - ( \alpha R_i^{(1)} + (1-\alpha) R_i^{(2)}) \\
&\triangleq \alpha  R_i^{(1)c} + (1-\alpha) R_i^{(2)c} \label{eq:convexity-2} \\
&\leq \log_2|\mathcal{F}| - H(N_0), \label{eq:convexity}
\end{align}
\end{subequations}
where \eqref{eq:convexity} follows from \eqref{eq:cap}.

From \eqref{eq:convexity-2} and \eqref{eq:cap}, we get
\begin{equation}
R_i^{(3)c} \leq \log_2|\mathcal{F}| - H(N_i).
\end{equation}
So, the rate tuple $(R_1^{(3)}, R_2^{(3)}, \dotsc, R_L^{(3)}) \in \mathcal{R}$, meaning that $\mathcal{R}$ is convex.

\subsection{The Common-Rate Capacity of the MWRC over a Finite Field}

Consider the common-rate case where all users transmit at the same rate, i.e., $R_i = R_\text{min}$, for all $i \in \{1,2,\dotsc,L\}$. We have $W_i = A_i$ and $B_i = \varnothing$, for all $i$, i.e., rate splitting is not required. So, using FDF, on the uplink, only the first $(L-1)$ sub-blocks are required for each message tuple for the users to transmit their respective $W_i$ in pairs. On the downlink, since $B_i = \varnothing$ for all $i$, the users do not need to use their own message in decoding $\mathbb{U}$ (c.f. \eqref{eq:down-mwrc-1}--\eqref{eq:down-mwrc-2}), i.e., joint decoding is not required. The users only utilize their respective messages in steps \eqref{eq:chain-decoding-start}--\eqref{eq:chain-decoding-stop} after they have decoded $\mathbb{U}$. FDF without rate splitting and separate source-channel decoding achieves the common-rate capacity, stated in the following corollary.
\begin{corollary}\label{corollary:capacity}
Consider the $L$-user MWRC over a finite field $\mathcal{F}$. The common-rate capacity is
\begin{equation}
C =  \frac{1}{L-1}\left(\log_2 |\mathcal{F}| - \max\limits_{i \in \{0,1,\dotsc,L\}} H(N_i)\right). \label{eq:capacity}
\end{equation}
\end{corollary}

\begin{proof}[Proof of Corollary~\ref{corollary:capacity}]
For the common-rate case, $R_i \triangleq R$, $\forall i\in \{1,2,\dotsc,L\}$ and we have $R_\text{min}^{\text{c}} = R_i^{\text{c}} = (L-1)R$, $\forall i$. From Theorem~\ref{theorem:mwrc-capacity}, all non-negative rate tuples $(R,R,\dotsc,R)$ satisfying
\begin{equation}
(L-1)R \leq \log_2|\mathcal{F}| - H(N_i), \quad \forall i \in \{0,1,\dotsc,L\},
\end{equation}
are achievable. So, common rates up to $\left(\log_2 |\mathcal{F}| - \max_{i \in \{0,1,\dotsc,L\}} H(N_i)\right)/(L-1)$ are achievable. From Corollary~\ref{theorem:upper-bound}, we know that this is a capacity upper bound.
\end{proof}

\section{A Case Study: The Binary Two-Way Relay Channel}\label{section:case-study}

In this section, we study the special case of the binary TWRC to illustrate the role of rate-splitting and joint source-channel decoding in achieving the capacity region. In the notation of this paper, we study the case where $L=2$, $\mathcal{F} = \{0,1\} \triangleq \mathcal{F}_2$, $\oplus$ and $\odot$ are addition and multiplication in modulo-two respectively. By definition, $h_{1,0} = h_{2,0} = h_{0,1} = h_{0,2} = 1$, since they cannot be zero. For the binary TWRC, the noise variables $N_0$, $N_1$, and  $N_2$ are each binary, and we can define $\rho_i \in [0,1]$ such that $\rho_i = \Pr\{N_i = 1\}$ and $H(\rho_i) = H(N_i) = -\rho_i\log_2\rho_i - (1-\rho_i)\log_2(1-\rho_i)$. Without loss of generality, we consider $\rho_i \in [0,\frac{1}{2}]$ for all $i \in \{0,1,2\}$. Although the capacity region of the binary TWRC has been reported in \cite{knopp07,namchung08}, we use this example to highlight the components of our scheme and to compare FDF with the complete-decode-forward (CDF) strategy.

\subsection{Functional-Decode-Forward with Rate Splitting and Joint Source-Channel Decoding}\label{section:fdf-rs-joint}
From Theorem~\ref{theorem:mwrc-capacity}, FDF with rate splitting and joint source-channel decoding achieves all non-negative rate pairs $(R_1,R_2)$ satisfying
\begin{align}
R_1,R_2 &< 1 - H(\rho_0)\label{eq:capacity-binary-1}\\
R_1 &< 1 - H(\rho_2)\label{eq:capacity-binary-2}\\
R_2 &< 1 - H(\rho_1)\label{eq:capacity-binary-3},
\end{align}
whose closure gives the capacity region. 

\subsection{Functional-Decode-Forward with Rate Splitting and Separate Source-Channel Decoding}

Now, we find the achievable rate region using FDF with rate splitting but with \emph{separate} source-channel decoding.

The coding on the uplink is the same as that in  Sec.~\ref{uplink}, i.e., using linear codes, functional decoding and rate splitting. First, we assume that $R_2 \geq R_1$, and hence $W_1 = A_1$ and $W_2 = (A_2,B_2)$. So, on the uplink, from \eqref{eq:uplink-3}, if $R_2 \leq 1 - H(\rho_0)$, then the relay can reliably decode $([\boldsymbol{S}(A_{1,2}^{(v)})]_{\forall v},[\boldsymbol{S}(B_2^{(v)})]_{\forall v})$.

Now, instead of using the joint source-channel decoding for the downlink described in Sec.~\ref{downlink}, we will use separate source-channel decoding in the sense that the users do not use their own messages in channel decoding. We re-cast the downlink as a \emph{broadcast channel with degraded message sets} \cite{kornermarton77}, where a source broadcasts a common message to two destinations and a private message to one of the destinations, and where both the destinations do not know the messages \emph{a priori}. Applying this to the downlink of the binary TWRC, we have the relay sending $[\boldsymbol{S}(A_{1,2}^{(v)})]_{\forall v}$ to both users, and $[\boldsymbol{S}(B_2^{(v)})]_{\forall v}$ to user 1, and the users do not use their own messages in the channel decoding of $[\boldsymbol{S}(A_{1,2}^{(v)})]_{\forall v}$ and $[\boldsymbol{S}(B_2^{(v)})]_{\forall v}$. 

Recall that $[\boldsymbol{S}(A_{1,2}^{(v)})]_{\forall v}$ is an $nR_1$-bit message and $[\boldsymbol{S}(B_2^{(v)})]_{\forall v}$ an $nR_2'$-bit message.
From \cite{kornermarton77}, if $R_1 < 1 - H\big(\beta(1-\rho_2) + (1-\beta)\rho_2\big)$, $R_2' < H\big(\beta(1-\rho_1) + (1-\beta)\rho_1\big) - H(\rho_1)$, and $R_1 + R_2' < 1 - H(\rho_1)$, for some $0 \leq \beta \leq \frac{1}{2}$, then both the users can reliably decode $[\boldsymbol{S}
(A_{1,2}^{(v)})]_{\forall v}$ and user 1 can reliably decode $[\boldsymbol{S}(B_2^{(v)})]_{\forall v}$ purely from their respective received signals $\boldsymbol{Y}_i$. Of course, after decoding $[\boldsymbol{S}(A_{1,2}^{(v)})]_{\forall v}$ and  $[\boldsymbol{S}(B_2^{(v)})]_{\forall v}$ (for user 1), the users must follow the steps in \eqref{eq:chain-decoding-start}--\eqref{eq:chain-decoding-stop} to obtain the other user's message. But as far as channel decoding on the downlink is concerned, the users' own messages are not used (as side information).

Combining the rate constraints on the uplink and on the downlink, we have the following achievable rate region:
\begin{theorem}\label{theorem:fdf-rs-separate}
Consider the two-user MWRC over $\mathcal{F}_2$. FDF with rate splitting and separate source-channel decoding achieves the convex hull of $\mathcal{R}_1$ and $\mathcal{R}_2$, where
\begin{itemize}
\item $\mathcal{R}_1$ is the set of all non-negative rate pairs $(R_1,R_1 + R_2')$ satisfying
\begin{align}
R_1 &< 1 - H\big(\beta(1-\rho_2) + (1-\beta)\rho_2\big) \label{eq:beta-r1}\\
R_2' &< H\big(\beta(1-\rho_1) + (1-\beta)\rho_1\big) - H(\rho_1)\\
R_1 + R_2' &< 1 - \max\{H(\rho_0), H(\rho_1)\},
\end{align}
for some $0 \leq \beta \leq \frac{1}{2}$.
\item $\mathcal{R}_2$ is the set of all non-negative rate pairs $(R_2 + R_1',R_2)$ satisfying
\begin{align}
R_2 &< - H\big(\alpha(1-\rho_1) + (1-\alpha)\rho_1\big) \label{eq:beta-r2}\\
R_1' &< H\big(\alpha(1-\rho_2) + (1-\alpha)\rho_2\big) - H(\rho_2)\\
R_2 + R_1' &< 1 - \max\{H(\rho_0), H(\rho_2)\},
\end{align}
for some $0 \leq \alpha \leq \frac{1}{2}$.
\end{itemize}
\end{theorem}

\begin{proof}[Proof of Theorem~\ref{theorem:fdf-rs-separate}]
$\mathcal{R}_1$ follows directly from the above-mentioned rate constraints. $\mathcal{R}_2$ is obtained by reversing the role of users 1 and 2 for the case $R_1 \geq R_2$. Using time sharing, the convex hull of $\mathcal{R}_1$ and $\mathcal{R}_2$ is achievable.
\end{proof}

\begin{remark}
We can show that when $\rho_1 \leq \rho_2$, $\mathcal{R}_2 \subseteq \mathcal{R}_1$; and vice versa. Hence, for any channel setting, it is sufficient to consider only one region in Theorem~\ref{theorem:fdf-rs-separate}.
\end{remark}

Now, we show that FDF with rate splitting and separate source-channel decoding achieves the capacity region of the binary TWRC under certain conditions.

\begin{lemma} \label{lemma:fdf-separate-optimal}
Consider the two-user MWRC over $\mathcal{F}_2$. If
\begin{enumerate}
\item $\rho_0 \geq \max \{ \rho_1, \rho_2 \}$, or
\item $\rho_1 = \rho_2$,
\end{enumerate}
then FDF with rate splitting and separate source-channel decoding achieves the capacity region.
\end{lemma}

\begin{proof}[Proof of Lemma~\ref{lemma:fdf-separate-optimal}]
First, consider the case $\rho_1 \leq \rho_2$, i.e., $H(\rho_2) \geq H(\rho_1)$. If
\begin{equation}
\rho_0 \geq \rho_2 \Leftrightarrow H(\rho_0) \geq H(\rho_2),
\end{equation}
we have
\begin{equation}
1 - H(\rho_0) \leq 1 - H(\rho_2) \leq 1 - H(\rho_1). \label{eq:noisy-up}
\end{equation}
Then by setting $\beta=0$, i.e., $R_2'=0$, $\mathcal{R}_1$ in Theorem~\ref{theorem:fdf-rs-separate} becomes
\begin{equation}
\{ (R_1,R_2): 0 \leq R_1, R_2 < 1 - H(\rho_0) \}. \label{eq:fdf-capacity}
\end{equation}
The closure of the above region coincides with the capacity region since \eqref{eq:capacity-binary-1} implies \eqref{eq:capacity-binary-2} and \eqref{eq:capacity-binary-3} when \eqref{eq:noisy-up} is true.

Similarly, for the case of $\rho_2 \leq \rho_1$, if $\rho_0 \geq \rho_1$, then the closure of $\mathcal{R}_2$ (with $\alpha=0$) in Theorem~\ref{theorem:fdf-rs-separate}  coincides with the capacity region.

Next, consider the case $\rho_1 = \rho_2$, i.e., $H(\rho_1)=H(\rho_2)$. By setting $\beta=0$, i.e., $R_2'=0$, $\mathcal{R}_1$ in Theorem~\ref{theorem:fdf-rs-separate} becomes
\begin{multline}
\{ (R_1,R_2): 0 \leq R_i < 1 - H(\rho_1), 0 \leq R_i < 1- H(\rho_0),\\ \text{ for } i = 1,2 \},
\end{multline}
whose closure also coincides with the capacity region.
\end{proof}

\subsection{Complete-Decode-Forward}
Using CDF, the relay fully decodes both $W_1$ (of $nR_1$ bits) and $W_2$ (of $nR_2$ bits) on the uplink, which is a multiple-access channel. So, if
\begin{align}
R_1 &< 1 - H(\rho_0)\label{eq:mac-1}\\
R_2 &< 1 - H(\rho_0)\label{eq:mac-2}\\
R_1 + R_2 &< 1 - H(\rho_0),\label{eq:mac-3}
\end{align}
then the relay can reliably decode $W_1$ and $W_2$~\cite{ahlswede71,liao72}. Note that \eqref{eq:mac-3} implies \eqref{eq:mac-1} and \eqref{eq:mac-2}.

Assuming that the relay has successfully decoded $W_1$ and $W_2$, it broadcasts $(W_1,W_2)$ on the downlink. Using joint source-channel decoding, each user $i$, $i\in \{1,2\}$, can reliably decode the other user's message from their respective received signals $\boldsymbol{Y}_i$ and their own messages $W_i$ if~\cite{kramershamai07,oechteringschnurr08}
\begin{align}
R_1 &< 1 - H(\rho_2)\\
R_2 &< 1 - H(\rho_1).
\end{align}

Combining the uplink and the downlink constraints, the achievable rate region using CDF is given by the following theorem:

\begin{theorem}\label{theorem:cdf}
Consider the two-user MWRC over $\mathcal{F}_2$. CDF achieves all non-negative rate pairs $(R_1,R_2)$ satisfying
\begin{align}
R_1 &< 1 - H(\rho_2)\label{eq:fdf-1} \\
R_2 &< 1 - H(\rho_1)\label{eq:fdf-2} \\
R_1 + R_2 &< 1 - H(\rho_0). \label{eq:fdf-3}
\end{align}
\end{theorem}

CDF achieves the capacity region under the following conditions.
\begin{lemma}\label{lemma:cdf-optimal}
Consider the two-user MWRC over $\mathcal{F}_2$. If
\begin{equation}
H(\rho_0) \leq H(\rho_1) + H(\rho_2) - 1,
\end{equation}
then CDF achieves the capacity region.
\end{lemma}

\begin{proof}[Proof of Lemma~\ref{lemma:cdf-optimal}]
\begin{align}
&H(\rho_0) \leq H(\rho_1) + H(\rho_2) - 1\\
\Leftrightarrow \quad &1 - H(\rho_0)  \geq 1 - H(\rho_1) + 1 - H(\rho_2) \label{eq:cdf-condition}\\
\Rightarrow \quad &H(\rho_1) \geq H(\rho_0) \text{ and } H(\rho_2) \geq H(\rho_0). \label{eq:cdf-condition-2}
\end{align}
From \eqref{eq:cdf-condition}, we know that conditions \eqref{eq:fdf-1} and \eqref{eq:fdf-2} imply \eqref{eq:fdf-3}. In this case, CDF achieves the following rate region
\begin{equation}
\{ (R_1,R_2): 0 \leq R_1 < 1 - H(\rho_2), 0 \leq R_2 < 1 - H(\rho_1) \},
\end{equation}
whose closure is the capacity region since \eqref{eq:cdf-condition-2}, \eqref{eq:capacity-binary-2} and \eqref{eq:capacity-binary-3} imply \eqref{eq:capacity-binary-1}.
\end{proof}

\subsection{Numerical Calculations and Discussion}

We denote FDF with rate splitting and joint source-channel decoding by FDF-RS (joint), and FDF with rate splitting and separate source-channel decoding by FDF-RS (separate) for the discussion in this section.

\begin{figure}[t]
\centering
\resizebox{8.8cm}{!}{
\begin{picture}(0,0)%
\includegraphics{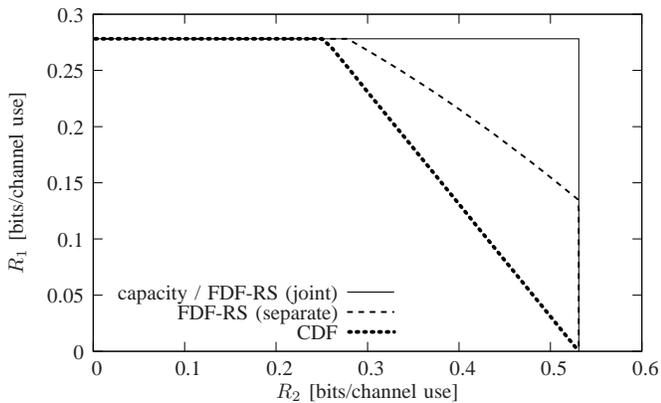}%
\end{picture}%
\setlength{\unitlength}{3947sp}%
\begingroup\makeatletter\ifx\SetFigFont\undefined%
\gdef\SetFigFont#1#2#3#4#5{%
  \reset@font\fontsize{#1}{#2pt}%
  \fontfamily{#3}\fontseries{#4}\fontshape{#5}%
  \selectfont}%
\fi\endgroup%
\begin{picture}(5220,3182)(1346,-3995)
\put(1963,-3623){\makebox(0,0)[rb]{\smash{{\SetFigFont{10}{12.0}{\familydefault}{\mddefault}{\updefault} 0}}}}
\put(1963,-3175){\makebox(0,0)[rb]{\smash{{\SetFigFont{10}{12.0}{\familydefault}{\mddefault}{\updefault} 0.05}}}}
\put(1963,-2726){\makebox(0,0)[rb]{\smash{{\SetFigFont{10}{12.0}{\familydefault}{\mddefault}{\updefault} 0.1}}}}
\put(1963,-2278){\makebox(0,0)[rb]{\smash{{\SetFigFont{10}{12.0}{\familydefault}{\mddefault}{\updefault} 0.15}}}}
\put(1963,-1830){\makebox(0,0)[rb]{\smash{{\SetFigFont{10}{12.0}{\familydefault}{\mddefault}{\updefault} 0.2}}}}
\put(1963,-1381){\makebox(0,0)[rb]{\smash{{\SetFigFont{10}{12.0}{\familydefault}{\mddefault}{\updefault} 0.25}}}}
\put(1963,-933){\makebox(0,0)[rb]{\smash{{\SetFigFont{10}{12.0}{\familydefault}{\mddefault}{\updefault} 0.3}}}}
\put(2038,-3748){\makebox(0,0)[b]{\smash{{\SetFigFont{10}{12.0}{\familydefault}{\mddefault}{\updefault} 0}}}}
\put(2767,-3748){\makebox(0,0)[b]{\smash{{\SetFigFont{10}{12.0}{\familydefault}{\mddefault}{\updefault} 0.1}}}}
\put(3496,-3748){\makebox(0,0)[b]{\smash{{\SetFigFont{10}{12.0}{\familydefault}{\mddefault}{\updefault} 0.2}}}}
\put(4226,-3748){\makebox(0,0)[b]{\smash{{\SetFigFont{10}{12.0}{\familydefault}{\mddefault}{\updefault} 0.3}}}}
\put(4955,-3748){\makebox(0,0)[b]{\smash{{\SetFigFont{10}{12.0}{\familydefault}{\mddefault}{\updefault} 0.4}}}}
\put(5684,-3748){\makebox(0,0)[b]{\smash{{\SetFigFont{10}{12.0}{\familydefault}{\mddefault}{\updefault} 0.5}}}}
\put(6413,-3748){\makebox(0,0)[b]{\smash{{\SetFigFont{10}{12.0}{\familydefault}{\mddefault}{\updefault} 0.6}}}}
\put(1481,-2216){\rotatebox{90.0}{\makebox(0,0)[b]{\smash{{\SetFigFont{10}{12.0}{\familydefault}{\mddefault}{\updefault}$R_1$ [bits/channel use]}}}}}
\put(4225,-3935){\makebox(0,0)[b]{\smash{{\SetFigFont{10}{12.0}{\familydefault}{\mddefault}{\updefault}$R_2$ [bits/channel use]}}}}
\put(3988,-3150){\makebox(0,0)[rb]{\smash{{\SetFigFont{10}{12.0}{\familydefault}{\mddefault}{\updefault}capacity / FDF-RS (joint)}}}}
\put(3988,-3309){\makebox(0,0)[rb]{\smash{{\SetFigFont{10}{12.0}{\familydefault}{\mddefault}{\updefault}FDF-RS (separate)}}}}
\put(3988,-3468){\makebox(0,0)[rb]{\smash{{\SetFigFont{10}{12.0}{\familydefault}{\mddefault}{\updefault}CDF}}}}
\end{picture}%
}
\caption{Rate region comparison for $1-H(\rho_0)=0.531$, $1-H(\rho_1)=0.714$, and $1-H(\rho_2)=0.278$}
\label{fig:comparison}
\end{figure}

In Fig.~\ref{fig:comparison}, we compare FDF-RS (joint), FDF-RS (separate), and CDF for the following channel parameters: $\rho_0=0.1$, $\rho_1=0.05$, and $\rho_2=0.2$. In this example, the FDF-RS (separate) achieves a rate region strictly larger than that of CDF, but both regions are strictly smaller than the capacity region which is achievable by FDF-RS (joint).

In Fig.~\ref{fig:capacity}, we fix $\rho_0=0.25$ and plot the range of $\rho_1$ and $\rho_2$ for which the capacity region is achieved by FDF-RS (separate) or CDF. The top-right corner corresponds to a noisier downlink ($\rho_1,\rho_2 > \rho_0$) , while the bottom-left corner to a noisier uplink ($\rho_0 > \rho_1, \rho_2$).

For the capacity region in Sec.~\ref{section:fdf-rs-joint}, we refer to the constraints \eqref{eq:capacity-binary-1} as the uplink constraints on the capacity region, and \eqref{eq:capacity-binary-2}-\eqref{eq:capacity-binary-3} the downlink constraints on the capacity region.

Using CDF, the relay needs to fully decode the users' messages on the uplink, and this restricts the sum rate to be constrained by the uplink, c.f. \eqref{eq:fdf-3}. When the uplink is noisy and is the channel bottleneck, the capacity region is effectively constrained by the uplink constraint \eqref{eq:capacity-binary-1}, which is strictly more relaxed than \eqref{eq:fdf-3}. So, CDF is not \emph{uplink optimized}.

However, when the downlink is noisy such that $H(\rho_0) \leq H(\rho_1) + H(\rho_2) - 1$, the capacity region is effectively constrained by the downlink constraints \eqref{eq:capacity-binary-2}-\eqref{eq:capacity-binary-3}, which is achievable by CDF, as shown in Lemma~\ref{lemma:cdf-optimal} and plotted in Fig.~\ref{fig:capacity}. We say that CDF is \emph{downlink optimized}.

\begin{figure}[t]
\centering
\resizebox{8.5cm}{!}{
\begin{picture}(0,0)%
\includegraphics{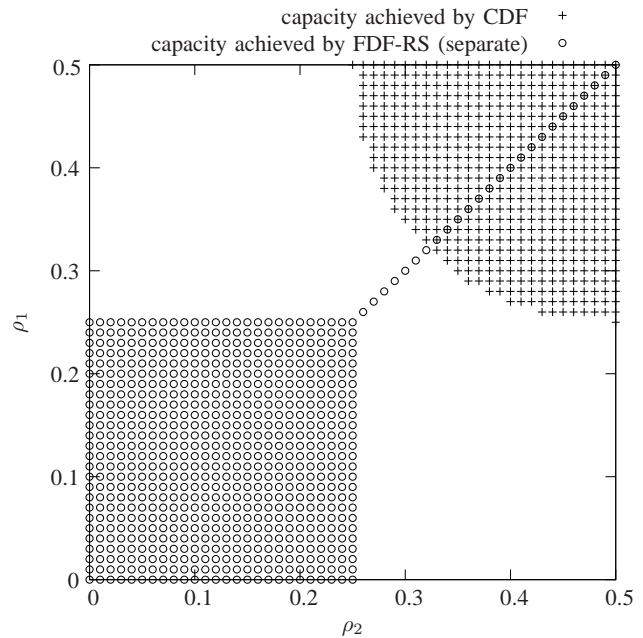}%
\end{picture}%
\setlength{\unitlength}{3947sp}%
\begingroup\makeatletter\ifx\SetFigFont\undefined%
\gdef\SetFigFont#1#2#3#4#5{%
  \reset@font\fontsize{#1}{#2pt}%
  \fontfamily{#3}\fontseries{#4}\fontshape{#5}%
  \selectfont}%
\fi\endgroup%
\begin{picture}(4620,4627)(1346,-3995)
\put(5168,509){\makebox(0,0)[rb]{\smash{{\SetFigFont{10}{12.0}{\familydefault}{\mddefault}{\updefault}capacity achieved by CDF}}}}
\put(1888,-3623){\makebox(0,0)[rb]{\smash{{\SetFigFont{10}{12.0}{\familydefault}{\mddefault}{\updefault} 0}}}}
\put(1888,-2869){\makebox(0,0)[rb]{\smash{{\SetFigFont{10}{12.0}{\familydefault}{\mddefault}{\updefault} 0.1}}}}
\put(1888,-2115){\makebox(0,0)[rb]{\smash{{\SetFigFont{10}{12.0}{\familydefault}{\mddefault}{\updefault} 0.2}}}}
\put(1888,-1361){\makebox(0,0)[rb]{\smash{{\SetFigFont{10}{12.0}{\familydefault}{\mddefault}{\updefault} 0.3}}}}
\put(1888,-607){\makebox(0,0)[rb]{\smash{{\SetFigFont{10}{12.0}{\familydefault}{\mddefault}{\updefault} 0.4}}}}
\put(1963,-3748){\makebox(0,0)[b]{\smash{{\SetFigFont{10}{12.0}{\familydefault}{\mddefault}{\updefault} 0}}}}
\put(2733,-3748){\makebox(0,0)[b]{\smash{{\SetFigFont{10}{12.0}{\familydefault}{\mddefault}{\updefault} 0.1}}}}
\put(3503,-3748){\makebox(0,0)[b]{\smash{{\SetFigFont{10}{12.0}{\familydefault}{\mddefault}{\updefault} 0.2}}}}
\put(4273,-3748){\makebox(0,0)[b]{\smash{{\SetFigFont{10}{12.0}{\familydefault}{\mddefault}{\updefault} 0.3}}}}
\put(5043,-3748){\makebox(0,0)[b]{\smash{{\SetFigFont{10}{12.0}{\familydefault}{\mddefault}{\updefault} 0.4}}}}
\put(5813,-3748){\makebox(0,0)[b]{\smash{{\SetFigFont{10}{12.0}{\familydefault}{\mddefault}{\updefault} 0.5}}}}
\put(1481,-1676){\rotatebox{90.0}{\makebox(0,0)[b]{\smash{{\SetFigFont{10}{12.0}{\familydefault}{\mddefault}{\updefault}$\rho_1$}}}}}
\put(3888,-3935){\makebox(0,0)[b]{\smash{{\SetFigFont{10}{12.0}{\familydefault}{\mddefault}{\updefault}$\rho_2$}}}}
\put(1770,137){\makebox(0,0)[b]{\smash{{\SetFigFont{10}{12.0}{\familydefault}{\mddefault}{\updefault} 0.5}}}}
\put(5168,307){\makebox(0,0)[rb]{\smash{{\SetFigFont{10}{12.0}{\familydefault}{\mddefault}{\updefault}capacity achieved by FDF-RS (separate)}}}}
\end{picture}%
}
\caption{This figure shows the regions of channel parameters $(\rho_1,\rho_2)$ for which the capacity region for $\rho_0=0.25$ is achieved by CDF and FDF-RS (separate). The capacity region for all $(\rho_1,\rho_2)$ can be achieved by FDF-RS (joint).}
\label{fig:capacity}
\end{figure}

Using FDF-RS (separate), the users' \emph{a priori} knowledge about their own messages is not utilized during the channel decoding on the downlink -- their own messages are used only \emph{after} channel decoding. So, FDF with separate source-channel decoding is not downlink optimized. This is why when the downlink is noisy ($\rho_1 > \rho_0$ or $\rho_2 > \rho_0$), FDF-RS (separate) fails to achieve the capacity region. An exception is when $\rho_1=\rho_2$, i.e., the downlink is \emph{symmetrical}, in this case, the equal rate point (common rate) marks a vertex of the capacity region and from Corollary~\ref{corollary:capacity}, we know that FDF with separate source-channel decoding achieves the common-rate capacity.

On the uplink, FDF-RS (separate) performs functional decoding at the relay and is able to achieve the uplink constraint on the capacity region. As shown in Lemma~\ref{lemma:fdf-separate-optimal} and plotted in Fig.~\ref{fig:capacity}, when the uplink is the channel bottleneck, FDF-RS (separate) achieves the capacity region.

From Fig.~\ref{fig:capacity}, we see that using both CDF and FDF-RS (separate) does not cover the capacity region for all channel settings. On the other hand, FDF-RS (joint) is both uplink and downlink optimized, and it achieves the capacity region for all channel settings.

\section{Conclusion}\label{section:conclusion}

We have proposed a functional-decode-forward (FDF) coding strategy with rate splitting and \emph{joint} source-channel decoding that achieves the capacity region of the multi-way relay channel (MWRC) over finite fields. For the special case where all users transmit at the same rate, our proposed FDF achieves the common-rate capacity of MWRCs over finite fields without requiring rate splitting or joint source-channel decoding.

Using the two-user binary MWRC as an example, we showed that both FDF with rate splitting and \emph{separate} source-channel decoding (denoted by FDF-RS (separate) in Figs.~\ref{fig:comparison} and~\ref{fig:capacity}), and complete-decode-forward (CDF) fail to achieve the capacity region of the MWRC as (i) for the former, users' messages are not utilized for channel decoding on the downlink and (ii) for the latter, the relay is constrained to decoding all users' messages. We noted that the shortcoming of CDF corresponds to the strength of FDF with rate splitting and separate source-channel decoding, and vice versa. However, as seen from Fig.~\ref{fig:capacity}, even considering both strategies does not cover the capacity region for all noise distributions.

Our proposed FDF with rate splitting and joint source-channel decoding overcomes these shortcomings by having the relay decode only functions of the source messages on the uplink, and having the users utilize their own messages in channel decoding on the downlink. This strategy indeed achieves the capacity regions of MWRCs over finite fields for all noise distributions. Our proposed coding strategy can be applied to the general multi-source multi-destination multi-relay network, where the relays facilitate data exchange among different source-destination pairs, but are themselves not required to decode the source messages.

\end{document}